\documentclass[dvipsnames,11pt]{article}
\usepackage{amsmath,amssymb,amsthm,color,bm,mathrsfs,extarrows,tikz,graphicx,mathtools,enumitem,fancybox}
\usepackage[noadjust]{cite}
\usepackage[ruled]{algorithm2e}
\usepackage[margin=1.in]{geometry}
\usepackage{xcolor,bbm}
\usepackage[pagebackref]{hyperref}
\usepackage{comment}

\setlist[itemize]{itemsep=0pt}
\setlist[enumerate]{itemsep=0pt}
\hypersetup{colorlinks=true,urlcolor=blue,linkcolor=blue,citecolor=[rgb]{.42,.56,.14},}
\usetikzlibrary{decorations.pathreplacing,decorations.markings,decorations.pathmorphing,decorations.shapes,arrows.meta,positioning}

\usepackage[capitalise,nameinlink]{cleveref}
\Crefname{lemma}{Lemma}{Lemmas}
\Crefname{fact}{Fact}{Facts}
\Crefname{theorem}{Theorem}{Theorems}
\Crefname{corollary}{Corollary}{Corollaries}
\Crefname{claim}{Claim}{Claims}
\Crefname{example}{Example}{Examples}
\Crefname{problem}{Problem}{Problems}
\Crefname{definition}{Definition}{Definitions}
\Crefname{notation}{Notation}{Notations}
\Crefname{assumption}{Assumption}{Assumptions}
\Crefname{subsection}{Subsection}{Subsections}
\Crefname{section}{Section}{Sections}
\Crefformat{equation}{(#2#1#3)}

\newtheorem{theorem}{Theorem}[section]
\newtheorem*{theorem*}{Theorem}

\newtheorem*{proposition*}{Proposition}
\newtheorem{lemma}[theorem]{Lemma}
\newtheorem*{lemma*}{Lemma}
\newtheorem{corollary}[theorem]{Corollary}
\newtheorem*{corollary*}{Corollary}
\newtheorem*{conjecture*}{Conjecture}
\newtheorem{fact}[theorem]{Fact}
\newtheorem*{fact*}{Fact}

\newtheorem*{exercise*}{Exercise}

\newtheorem*{hypothesis*}{Hypothesis}
\newtheorem{conjecture}[theorem]{Conjecture}

\theoremstyle{definition}
\newtheorem{definition}[theorem]{Definition}
\newtheorem{notation}[theorem]{Notation}

\newtheorem{example}[theorem]{Example}

\newtheorem{exercise-easy}[theorem]{Exercise}
\newtheorem{exercise-med}[theorem]{Exercise}
\newtheorem{exercise-hard}[theorem]{Exercise$^\star$}
\newtheorem{claim}[theorem]{Claim}
\newtheorem*{claim*}{Claim}

\newtheorem*{remark*}{Remark}

\newtheorem*{observation*}{Observation}

\usepackage{prettyref}
\newcommand{\savehyperref}[2]{\texorpdfstring{\hyperref[#1]{#2}}{#2}}

\newrefformat{eq}{\savehyperref{#1}{\textup{(\ref*{#1})}}}
\newrefformat{ineq}{\savehyperref{#1}{\textup{(\ref*{#1})}}}
\newrefformat{lem}{\savehyperref{#1}{Lemma~\ref*{#1}}}
\newrefformat{def}{\savehyperref{#1}{Definition~\ref*{#1}}}
\newrefformat{thm}{\savehyperref{#1}{Theorem~\ref*{#1}}}
\newrefformat{cor}{\savehyperref{#1}{Corollary~\ref*{#1}}}
\newrefformat{cha}{\savehyperref{#1}{Chapter~\ref*{#1}}}
\newrefformat{sec}{\savehyperref{#1}{Section~\ref*{#1}}}
\newrefformat{app}{\savehyperref{#1}{Appendix~\ref*{#1}}}
\newrefformat{tab}{\savehyperref{#1}{Table~\ref*{#1}}}
\newrefformat{fig}{\savehyperref{#1}{Figure~\ref*{#1}}}
\newrefformat{hyp}{\savehyperref{#1}{Hypothesis~\ref*{#1}}}
\newrefformat{alg}{\savehyperref{#1}{Algorithm~\ref*{#1}}}
\newrefformat{rem}{\savehyperref{#1}{Remark~\ref*{#1}}}
\newrefformat{item}{\savehyperref{#1}{Item~\ref*{#1}}}
\newrefformat{step}{\savehyperref{#1}{step~\ref*{#1}}}
\newrefformat{conj}{\savehyperref{#1}{Conjecture~\ref*{#1}}}
\newrefformat{fact}{\savehyperref{#1}{Fact~\ref*{#1}}}
\newrefformat{fac}{\savehyperref{#1}{Fact~\ref*{#1}}}
\newrefformat{prop}{\savehyperref{#1}{Proposition~\ref*{#1}}}
\newrefformat{prob}{\savehyperref{#1}{Problem~\ref*{#1}}}
\newrefformat{claim}{\savehyperref{#1}{Claim~\ref*{#1}}}
\newrefformat{relax}{\savehyperref{#1}{Relaxation~\ref*{#1}}}
\newrefformat{rem}{\savehyperref{#1}{Remark~\ref*{#1}}}
\newrefformat{red}{\savehyperref{#1}{Reduction~\ref*{#1}}}
\newrefformat{part}{\savehyperref{#1}{Part~\ref*{#1}}}
\newrefformat{ex}{\savehyperref{#1}{Exercise~\ref*{#1}}}
\newrefformat{property}{\savehyperref{#1}{Property~\ref*{#1}}}
\newrefformat{case}{\savehyperref{#1}{Case~\ref*{#1}}}
\newrefformat{cat}{\savehyperref{#1}{Category~\ref*{#1}}}
\newrefformat{obs}{\savehyperref{#1}{Observation~\ref*{#1}}}
\newrefformat{cond}{\savehyperref{#1}{Condition~\ref*{#1}}}
\newrefformat{que}{\savehyperref{#1}{Question~\ref*{#1}}}

\DeclareSymbolFont{largesymbolsyhmath}{OMX}{yhex}{m}{n}
\DeclareMathAccent{\widehat}{\mathord}{largesymbolsyhmath}{"62}
\DeclareMathAccent{\widetilde}{\mathord}{largesymbolsyhmath}{"65}

\DeclareMathOperator*{\E}{\mathbb E}

\renewcommand{\Pr}{\operatorname*{\mathbf{Pr}}}

\newcommand{\eps}{\varepsilon}
\newcommand{\abs}[1]{\left| #1 \right|}
\newcommand{\vabs}[1]{\left\| #1 \right\|}
\newcommand{\abra}[1]{\left\langle #1 \right\rangle}
\newcommand{\pbra}[1]{\left( #1 \right)}
\newcommand{\sbra}[1]{\left[ #1 \right]}
\newcommand{\cbra}[1]{\left\{ #1 \right\}}
\renewcommand{\mid}{\,\middle\vert\,}
\newcommand{\bin}{\{0,1\}}
\newcommand{\binpm}{\{\pm1\}}
\newcommand{\poly}{\mathsf{poly}}
\newcommand{\polylog}{\mathsf{polylog}}
\newcommand{\Span}{\mathsf{Span}}
\newcommand{\rank}{\mathsf{rank}}
\newcommand{\sgn}{\mathsf{sgn}}
\newcommand{\indicator}{\mathsf{1}}
\newcommand{\Enc}{\mathsf{Enc}}
\newcommand{\depth}{\mathsf{depth}}

\newcommand{\ubm}{\bm{u}}
\newcommand{\vbm}{\bm{v}}

\newcommand{\Fbb}{\mathbb{F}}
\newcommand{\Nbb}{\mathbb{N}}
\newcommand{\Rbb}{\mathbb{R}}

\newcommand{\Dcal}{\mathcal{D}}
\newcommand{\Ecal}{\mathcal{E}}
\newcommand{\Fcal}{\mathcal{F}}
\newcommand{\Hcal}{\mathcal{H}}
\newcommand{\Lcal}{\mathcal{L}}

\newcommand{\Pcal}{\mathcal{P}}
\newcommand{\Scal}{\mathcal{S}}
\newcommand{\Tcal}{\mathcal{T}}
\newcommand{\Ucal}{\mathcal{U}}

\newcommand{\Xcal}{\mathcal{X}}

\renewcommand{\tilde}{\widetilde}

\renewcommand{\hat}{\widehat}

\title{Fourier Growth of Parity Decision Trees}
\author{
Uma Girish\thanks{Department of Computer Science, Princeton University. Email: \texttt{ugirish@cs.princeton.edu}}
\and
Avishay Tal\thanks{Department of EECS, University of California at Berkeley. Email: \texttt{avishay.tal@gmail.com}}
\and
Kewen Wu\thanks{Department of EECS, University of California at Berkeley. Email: \texttt{shlw\_kevin@hotmail.com}}
}
\date{}

\begin{document}
\maketitle

\begin{abstract}
We prove that for every parity decision tree of depth $d$ on $n$ variables, the sum of absolute values of Fourier coefficients at level $\ell$ is at most $d^{\ell/2} \cdot O(\ell \cdot \log(n))^\ell$. 
Our result is nearly tight for small values of $\ell$ and extends a previous Fourier bound for standard decision trees by Sherstov, Storozhenko, and Wu (STOC, 2021). 

As an application of our Fourier bounds, using the results of Bansal and Sinha (STOC, 2021), we show that the $k$-fold Forrelation problem has (randomized) parity decision tree complexity  $\tilde{\Omega}\left(n^{1-1/k}\right)$, while having quantum query complexity $\lceil k/2\rceil$. 

Our proof follows a random-walk approach, analyzing the contribution of a random path in the decision tree to the level-$\ell$ Fourier expression. 
To carry the argument, we apply a careful cleanup procedure to the parity decision tree, ensuring that the value of the random walk is bounded with high probability. We observe that step sizes for the level-$\ell$ walks can be computed by the intermediate values of level $\le \ell-1$ walks, which calls for an inductive argument.
Our approach differs from previous proofs of Tal (FOCS, 2020) and Sherstov, Storozhenko, and Wu (STOC, 2021) that relied on decompositions of the tree. In particular, for the special case of standard decision trees we view our proof as slightly simpler and more intuitive.

In addition, we prove a similar bound for noisy decision trees of cost at most $d$ -- a model that was recently introduced by Ben-David and Blais (FOCS, 2020).
\end{abstract}

\newpage
\section{Introduction}\label{sec:intro}
A common theme in the analysis of Boolean functions is proving structural results on classes of Boolean devices (e.g., decision trees, bounded-depth circuits) and then exploiting the structure to: (i) devise pseudorandom generators fooling these devices, (ii) prove lower bounds, showing that some explicit function cannot be computed by such Boolean devices of certain size, or (iii) design learning algorithms for the class of Boolean devices in either the membership-query model or the random-samples model.
Such structural results can involve properties of the Fourier spectrum of Boolean functions associated with Boolean devices, like concentration on low-degree terms or concentration on a few terms (i.e., ``approximate sparsity'').

In this work, we investigate the Fourier spectrum of parity decision trees. A parity decision tree (PDT) is an extension of the standard decision tree model. A PDT is a binary tree where each internal node is marked by a linear function (modulo $2$) on the input variables $(x_1, \ldots, x_n)$,  with two outgoing edges marked with $0$ and $1$, and each leaf is marked with either $0$ or $1$. A PDT naturally describes a computational model: on input $x = (x_1, \ldots, x_n)$, start at the root and at each step query the linear function specified by the current node on the input $x$ and continue on the edge marked with the value of the linear function evaluated on $x$. Finally, when reaching a leaf, output the value specified in the leaf. PDTs naturally generalize standard decision trees that can only query the value of a single input bit in each internal node.

PDTs were introduced in the seminal paper of Kushilevitz and Mansour \cite{DBLP:journals/siamcomp/KushilevitzM93}. Aligned with the aforementioned theme, Kushilevitz and Mansour proved a structural result for PDTs and used it to design learning algorithms for PDTs. They showed that every PDT of size $s$ computing a Boolean function $f\colon\bin^n \to \bin$ has 
$$
L_1(f) \triangleq \sum_{S \subseteq [n]} \abs{\hat{f}(S)} \le s,
$$
where $\hat{f}(S)$ are the Fourier coefficients of $f$ (see \Cref{sec:Fourier} for a precise definition). Then, they gave a learning algorithm in the membership-query model, running in time $\poly(t,n)$ that can learn any function $f$ with $L_1(f) \le t$. Combining the two results together, they obtained a $\poly(s,n)$-time algorithm for learning PDTs of size $s$.

Parity decision trees were also studied in relation to communication complexity and the log-rank conjecture \cite{MO09, ZS09, ZS10, TWXZ13, STV17, OWZST14, CS16,KQS15,DBLP:journals/siamcomp/HatamiHL18,Sanyal19, MS20
}.
Suppose Alice gets input $x \in \bin^n$, Bob gets input $y\in \bin^n$ and they want to compute some function $f(x,y)$. When $f$ is an XOR-function, namely $f(x,y) = g(x\oplus y)$ for some $g: \bin^n \to \bin$, then any PDT for $g$ of depth $d$ can be translated into a communication protocol for $f$ at cost $2d$: Alice and Bob simply traverse the PDT together, both exchanging the parity of their part of the input to simulate each query in the PDT. 
With this view, parity decision trees can be thought of as special cases of communication protocols for XOR functions.
A surprising result by Hatami, Hosseini, and Lovett \cite{DBLP:journals/siamcomp/HatamiHL18}, shows that this is not far from the optimal strategy for XOR functions. Namely, if the communication cost for computing $f$ is $c$, then the parity decision tree complexity of $g$ is at most $\poly(c)$. Due to this connection, the log-rank conjecture for XOR-functions reduces to the question of whether Boolean functions with at most $s$ non-zero Fourier coefficients can be computed by PDTs of depth $\polylog(s)$ \cite{MO09, ZS09}. The best known upper bound is that such functions can be computed by PDTs of depth $O(\sqrt{s})$ \cite{TWXZ13} (or even non-adaptive PDTs of depth $\tilde{O}(\sqrt{s})$  \cite{Sanyal19}).

While having small $L_{1}(f)$ norm implies learning algorithms and also simple pseudorandom generators fooling $f$~\cite{DBLP:journals/siamcomp/NaorN93}, this property can be quite restrictive. In particular, very simple functions (e.g., the Tribes function) have $L_{1}(f)$ exponential in $n$. Such examples motivated Reingold, Steinke, and Vadhan \cite{DBLP:conf/approx/ReingoldSV13} to study a more refined notion measuring for a given level $\ell$, the sum of absolute values of Fourier coefficients of sets $S$ of size exactly $\ell$, i.e, to study 
$$
L_{1,\ell}(f) \triangleq \sum_{S\subseteq [n]:|S|=\ell} \abs{\hat{f}(S)}.
$$
In particular, for $\ell=1$, the measure $L_{1,1}(f)$ is tightly related to the total influence of $f$ (and equals to it if $f$ is monotone).
The idea behind this more refined notion is that Fourier coefficients of different levels behave differently under standard manipulations to the function like random restrictions or noise operators. For example, when applying a noise operator with parameter $\gamma$, level-$\ell$ coefficients are multiplied by $\gamma^\ell$.
This motivates to establish a bound of the form
$L_{1,\ell}(f) \le t^{\ell}$ for some parameter $t$ and all $\ell = 1, \ldots, n$.  If $f$ satisfies such a bound, we say that $f \in \Lcal_{1}(t)$.\footnote{Note that if $f\in \Lcal_1(t)$ then after applying noise operator with $\gamma = 1/(2t)$, the noisy-version of $f$ has total $L_1$-norm at most $O(1)$ which makes it is quite easy to fool using small-biased distributions \cite{DBLP:journals/siamcomp/NaorN93}.}

Reingold, Steinke, and Vadhan \cite{DBLP:conf/approx/ReingoldSV13} showed that for read-once permutation branching programs of width $w$, while $L_{1}(f)$ could be exponential in $n$ (even for $w=3$), it nevertheless holds that
$L_{1,\ell}(f) \le (2w^2)^{\ell}$ for all $\ell=1, \ldots, n$. 
Then, they constructed a pseudorandom generator that fools any class of read-once branching programs for which $f \in \Lcal_1(t)$ using only $t \cdot \polylog(n)$ random bits.
This result was significantly generalized to a pseudorandom generator that fools any class of functions $f \in \Lcal_1(t)$ using only $t^2 \cdot \polylog(n)$ random bits~\cite{DBLP:journals/toc/ChattopadhyayHH19}. 
Further results established pseudorandom generators assuming  $L_{1,\ell}$ bounds only on the first few levels  \cite{DBLP:conf/stoc/ChattopadhyayHR18,DBLP:journals/corr/abs-2008-01316}.

It turns out that read-once permutation branching programs are just one example of many well-studied Boolean devices with non-trivial $L_{1,\ell}$ bounds. The following classes of Boolean functions are other examples:
\begin{enumerate}
	\item Width-$w$ CNF and width-$w$ DNF formulae are in $\Lcal_{1}(O(w))$ \cite{DBLP:journals/jcss/Mansour95}.
	\item $\mathsf{AC}^0$ circuits of size $s$ and depth $d$ are in $\Lcal_{1}\pbra{O(\log(s))^{d-1}}$ \cite{DBLP:conf/coco/Tal17}.
	\item Boolean functions with max-sensitivity at most $s$ are in $\Lcal_{1} (O(s))$ \cite{DBLP:journals/eccc/GopalanSTW16}
	\item Read-once branching programs of width $w$ are in $\Lcal_{1}\pbra{O(\log(n))^w}$\cite{DBLP:conf/stoc/ChattopadhyayHR18}
	\item Deterministic and randomized decision trees of depth $d$ are in $\Lcal_{1}\pbra{O\pbra{\sqrt{d \log(n)}}}$ \cite{DBLP:conf/focs/Tal20,  DBLP:journals/eccc/SherstovSW20}.
	\item If $f(x,y)$ is a function computed by communication protocol exchanging at most $c$ bits, then $h(z) = \E_x[f(x,x \oplus z)]$ satisfies $h \in \Lcal_1(O(c))$	\cite{DBLP:conf/innovations/GirishRT21,DBLP:journals/eccc/GirishRZ20a}.
	\item Polynomials $f$ over $\mathsf{GF}(2)$ of degree $d$ have $L_{1,\ell}(f) \le \left(2^{3d} \cdot \ell\right)^\ell$ \cite{DBLP:journals/toc/ChattopadhyayHH19}.
	\item Product tests, i.e., the XOR of multiple Boolean functions operating on  disjoint sets of at most $m$ bits each, are in $\Lcal_1(O(m))$ \cite{DBLP:conf/coco/Lee19}.
\end{enumerate}
We remark that Items~1, 2, 4, 5 and 8 are essentially tight, Item~3 can be potentially improved polynomially \cite{DBLP:journals/corr/abs-1204-6447,DBLP:journals/siamcomp/ODonnellS07}, Item~6 can be potentially improved quadratically \cite{DBLP:conf/innovations/GirishRT21} and Item~7 can be potentially improved exponentially \cite{DBLP:conf/innovations/ChattopadhyayHL19}. Indeed, improving Item~7 exponentially would imply that $\mathsf{AC}^0[\oplus]$ in $\Lcal_1(\polylog(n))$ and would give the first poly-logarithmic pseudorandom generators for this well-studied class of Boolean circuits \cite{DBLP:conf/innovations/ChattopadhyayHL19}.

The most relevant result to our work is the recent tight bounds on the $L_{1,\ell}$ of decision trees of depth $d$. Sherstov, Storozhenko and Wu \cite{DBLP:journals/eccc/SherstovSW20} recently proved that for any randomized decision tree of depth $d$ computing a function $f$, it holds that   $L_{1,\ell}(f) \le \sqrt{ {\binom d\ell} \cdot O(\log(n))^{\ell-1}}$. Their bound is nearly tight (see~\cite[Section~7]{DBLP:conf/focs/Tal20} and~\cite[Chapter~5.3]{DBLP:books/daglib/0033652} for tightness examples). One motivation for showing such a bound for decision trees is that it demonstrates a stark difference between quantum algorithms making few queries and randomized algorithms making a few queries. Indeed, the Fourier spectrum associated with quantum query algorithms making a few queries can be far from being approximately sparse (in the sense that its $L_{1,\ell}$ is quite large).
Based on that difference, both \cite{DBLP:journals/eccc/SherstovSW20} and \cite{DBLP:journals/eccc/BansalS20} showed that there are partial functions, either $k$-fold Forrelation or $k$-fold Rorrelation, that can be correctly computed with probability at least $1/2 + \Omega(1)$ by quantum algorithms making $\lceil{k/2\rceil}$ queries, but require $\tilde{\Omega}\pbra{n^{1-1/k}}$ queries for any randomized algorithm.
Moreover, due to the result of Aaronson and Ambainis \cite{DBLP:journals/siamcomp/AaronsonA18} this is the largest possible separation between the two models.

Indeed, as suggested in \cite{DBLP:conf/focs/Tal20}, one can show that any function with sufficiently good bounds on its $L_{1,\ell}$, for all $\ell = 1, \ldots, n$, cannot solve the $k$-fold Rorrelation, and such bounds were obtained by \cite{DBLP:journals/eccc/SherstovSW20} for randomized decision trees of depth $n^{1-1/k}/\polylog(n)$.
Independently, Bansal and Sinha obtained the same separation but only relying on the $L_{1,\ell}$ bounds for $\ell \in \{k, k+1, \ldots, k^2\}$. With this additional flexibility, they were able to obtain their separation for the simpler and explicit function called $k$-fold Forrelation.

For parity decision trees, the work of Blais, Tan, and Wan \cite{DBLP:journals/corr/BlaisTW15} established a tight bound of $O\pbra{\sqrt{d}}$ on the first level $\ell=1$. To the best of our knowledge, bounds on higher levels were not considered previously in the literature (in fact, even for standard decision trees, such bounds were not considered prior to \cite{DBLP:conf/focs/Tal20}).

\subsection{Our Results}
We prove level-$\ell$ bounds for any parity decision tree of depth $d$.
\begin{theorem}[Informal]\label{thm:informal}
Let $\Tcal$ be a depth-$d$ parity decision tree on $n$ variables. Then the sum of absolute Fourier coefficients at level $\ell$ is bounded by  $d^{\ell/2}\cdot O\!\pbra{\ell\cdot\log(n)}^\ell$.
\end{theorem}

See \Cref{thm:level_1} and \Cref{thm:level_l} for a precise statement taking into account the probability that $\Tcal$ accepts a uniformly random input. 
\Cref{thm:informal} extends the result of \cite{DBLP:journals/eccc/SherstovSW20} from standard decision trees to parity decision trees at the cost of an $(\ell \cdot \log(n))^{O(\ell)}$ multiplicative factor. We remark that even for standard decision tree there is a lower bound of $L_{1,\ell}(f) \ge \sqrt{ \binom{d}{\ell} \cdot (\log(n))^{\ell-1}}$ \cite[Section~7]{DBLP:conf/focs/Tal20} for constant $\ell$ and  $L_{1,\ell}(f) \ge \frac{1}{\poly(\ell)} \cdot \sqrt{ \binom{d}{\ell}}$ for all $\ell$ \cite[Chapter~5.3]{DBLP:books/daglib/0033652}. Thus, our bounds are tight up to $\polylog(n)$ factors for constant $\ell$, and they deteriorate as $\ell$ grows. Nevertheless, our main application relies on the bounds for small values of $\ell$ (constant or at most $\log^2 n$).

\paragraph{Noisy Decision Trees.}
We also investigate the Fourier spectrum of noisy decision trees.
Noisy decision trees are a different generalization of the standard model; here in each internal node $v$ we query a noisy version of an input bit, that equals the true bit with probability $(1+\gamma_v)/2$. Any such query costs $\gamma_v^2$.  We say that a noisy decision tree has cost at most $d$ if the total cost in any root-to-leaf path is at most $d$.
Recent work studied this model and established connections to the question of how randomized decision tree complexity behaves under composition \cite{DBLP:conf/focs/Ben-DavidB20}. 

We prove level-$\ell$ bounds for any noisy decision tree of cost at most $d$. See \Cref{thm:level-l_noisy_DT} for a precise statement.
\begin{theorem}[Informal]\label{thm:noisy-informal}
Let $\Tcal$ be a noisy decision tree of cost at most $d$ on $n$ variables. Then the sum of absolute Fourier coefficients at level $\ell$ is bounded by  $O(d)^{\ell/2}\cdot \pbra{\ell\cdot\log(n)}^{(\ell-1)/2}$.
\end{theorem}

\paragraph{Extension to Randomized Query Models.}
It is simple to verify that if $f$ is a convex combination of Boolean functions $f_1, \ldots, f_m$ each with $L_{1,\ell}(f_i) \le t_{\ell}$ then also $f$ satisfy $L_{1,\ell}(f) \le t_{\ell}$.
Thus, if we take a distribution over PDTs of depth $d$ (resp., noisy decision trees of cost $d$) we get the same bounds on their $L_{1,\ell}$ as those in \Cref{thm:informal} (resp., \Cref{thm:noisy-informal}). This is captured in the following corollary.

\begin{corollary}\label{cor:RPDT}
	Let $\Tcal$ be a randomized parity decision tree of depth at most $d$ on $n$ variables.
	Then, $$\forall \ell \in [n]: L_{1,\ell}(\Tcal) \le d^{\ell/2} \cdot O(\ell \cdot \log(n))^\ell.$$
	
	Let $\Tcal'$ be a randomized noisy decision tree of cost at most $d$ on $n$ variables.
	Then, $$\forall \ell \in [n]: L_{1,\ell}(\Tcal') \le O(d)^{\ell/2} \cdot (\ell \cdot \log(n))^{(\ell-1)/2}.$$
\end{corollary}

\subsection{Applications}\label{sec:applications}

\paragraph{Quantum versus Randomized Query Complexity.}
Let $k \le \log(n)$. Bansal and Sinha~\cite{DBLP:journals/eccc/BansalS20} gave a $\lceil k/2\rceil$ versus $\tilde{\Omega}\pbra{n^{1-1/k}}$ separation between the quantum and randomized query complexity of $k$-fold Forrelation (defined by \cite{DBLP:journals/siamcomp/AaronsonA18}).
For our purposes just think of $k$-fold Forrelation as a partial Boolean function on $n$ input bits.
Our main application is an extension of Bansal and Sinha's lower bound for the model of randomized parity decision trees.
This follows from their main technical result and \Cref{thm:informal}.
\begin{theorem}[Restatement of \protect{\cite[Theorem~3.2]{DBLP:journals/eccc/BansalS20}}]\label{thm:BS}
Let $f\colon\bin^{n} \to [0,1]$ such that $f$ and all its restrictions satisfy $L_{1,\ell}(f) \le t^{\ell}$ for $\ell = \{k, \ldots, k(k-1)\}$.
Let $\delta = 2^{-5k}$. 
Suppose $f$ is $\delta$-close to the value of $k$-fold Forrelation of $x$ for all $x$ on which $k$-fold Forrelation is defined.
Then, $t \ge \Omega\left(\frac{n^{(1-1/k)/2}}{k^{15}}\right)$.
\end{theorem}
\begin{corollary}
If $\Tcal$ is a randomized parity decision tree of depth $d$ computing $k$-fold Forrelation with success probability $\frac{1}{2}+\gamma$, then $d \ge {\gamma^2}\cdot \frac{n^{1-1/k}}{\poly(k) \log^2 n}.$
\end{corollary}
\begin{proof}
We can amplify the success probability of the randomized parity decision tree from $1/2 + \gamma$ to $1- 2^{-5k}$  by repeating the query algorithm $O(k/\gamma^2)$ times independently and taking majority. This results in a randomized parity decision tree $\Tcal'$ of depth $d' = O(d \cdot k/\gamma^2)$.
Now, \Cref{cor:RPDT} gives $L_{1,\ell}(\Tcal') \le (d')^{\ell/2} \cdot \allowbreak O(\ell \cdot \log(n))^\ell$ for all $\ell$.
In particular, 
$L_{1,\ell}(\Tcal') \le t^{\ell}$ for all $\ell \le k(k-1)$
where $t = O\pbra{\sqrt{d'} \cdot k(k-1) \cdot \log(n)}$.
This is also true for any restriction of $\Tcal'$, since fixing variables to constants yields another randomized parity decision tree of depth at most $d'$.
Combining the bounds on $L_{1,\ell}(\Tcal')$ for $\ell\in \{k, \ldots, k(k-1)\}$ with Theorem~\ref{thm:BS} gives 
$d' \ge \frac{n^{1-1/k}}{O(k^{34})  \cdot \log^2(n)}$ and thus $d \ge  {\gamma^2}\cdot \frac{n^{1-1/k}}{O(k^{35})  \cdot \log^2(n)}$.
\end{proof}

For constant $k$ and $\gamma = 2^{-O(k)}$, we get a $\lceil{k/2\rceil}$ versus $\tilde{\Omega}\pbra{n^{1-1/k}}$ separation between the quantum query complexity and the randomized parity query complexity of $k$-fold Forrelation. We remark that separations in the reverse direction are also known: for the $n$-bit parity function, the (randomized) parity query complexity is $1$ whereas the quantum query complexity is $\Omega(n)$ \cite{DBLP:journals/tcs/MontanaroNR11}.

Similarly, we can obtain the following corollary for noisy decision trees.
\begin{corollary}
If $\Tcal$ is a randomized noisy decision tree of cost at most $d$ computing $k$-fold Forrelation with success probability $\frac{1}{2}+\gamma$, then $d \ge \gamma^2\cdot \frac{n^{1-1/k}}{\poly(k) \log(n)}.$
\end{corollary}

\paragraph{Towards Communication Complexity Lower Bounds.}
We recall an open question from \cite{DBLP:conf/innovations/GirishRT21}, which, if true, would demonstrate that the randomized communication complexity of the Forrelation problem composed with the XOR gadget is $\tilde\Omega(n^{1/2})$. The \emph{simultaneous} quantum communication complexity of this problem is $O(\polylog(n))$ and the best known randomized lower bound is  $\tilde\Omega(n^{1/4})$ due to~\cite{DBLP:conf/innovations/GirishRT21}.  
\begin{conjecture}\label{conj:GRT}
Let $f\colon \bin^n \times \bin^n \to \bin$ computed by a deterministic communication protocol of cost at most $c$.
Let $h: \bin^n \to [0,1]$ defined by $h(z) = \E_x[f(x, x\oplus z)]$.
Then, $L_{1,2}(h) \le c \cdot \polylog(n)$.
\end{conjecture}
We view \Cref{thm:informal} as a first step towards this conjecture.
Indeed, for communication protocols that follow a parity decision tree strategy according to some tree $\Tcal$, it is simple to verify that $h = \Tcal$ (as functions), and thus 
$L_{1,2}(h) = L_{1,2}(\Tcal) \le c\cdot \polylog(n)$.

We remark that there is a separation of $O(\polylog(n))$ versus $\tilde\Omega(n^{1/2})$ between \emph{simultaneous} quantum communication complexity and   \emph{two-way} randomized communication complexity due to~\cite{gavinsky}. We also know a separation of $O(k\log n)$ versus $\tilde\Omega(n^{1-1/k})$  between \emph{two-way} quantum communication complexity and \emph{two-way} randomized communication complexity. This can be obtained by combining the optimal quantum versus classical query complexity separations of ~\cite{DBLP:journals/eccc/BansalS20} and ~\cite{DBLP:journals/eccc/SherstovSW20} and the query-to-communication lifting theorems \cite{DBLP:conf/icalp/ChattopadhyayFK19} using the inner product gadget.

\paragraph{Application to Expander Random Walk.}
Recently, \cite{DBLP:journals/eccc/CohenPT20} showed that expander random walks fool symmetric functions and also general functions in $\Lcal_1(t)$. To be more precise, assume $f\in\Lcal_1(t)$. Let $G$ be an expander, with second eigenvalue $\lambda \ll \frac{1}{t^4}$, where half of $G$'s vertices are labeled by $0$ and the rest are labeled by $1$. Then the expected value of $f$ on bits sampled by an $(m-1)$-step random walk on $G$ is approximately the value it would get on a uniformly random string in $\bin^m$. Combined with our results, this shows that if $f$ can be computed by low-depth parity decision trees then $f$ can be fooled by the expander random walk.

\paragraph{Fourier Bounds for Small-size Parity Decision Trees.}
By a simple size-to-depth reduction we obtain Fourier bounds for parity decision trees of bounded size.
We defer the simple proof to \Cref{sec:small-size decision trees}.
\begin{corollary}\label{cor:small-size}
Let $\Tcal$ be a parity decision tree of size at most $s>1$ on $n$ variables.
Then, 
$$\forall{\ell \in [n]}: L_{1,\ell}(f) \le (\log(s))^{\ell/2} \cdot O(\ell \cdot \log(n))^{1.5\ell}.$$
\end{corollary}

\subsection{Technical Overview}\label{sec:proof_overview}
For the rest of the paper we consider Boolean functions as functions from $\binpm^n$ to $\bin$. This is for convenience, since most of our calculations become easier under this representation. Observe that under this view, a parity decision tree  queries at each internal node the product $\prod_{i\in S} x_i$ for some $S \subseteq [n]$ and goes left/right depending on whether $\prod_{i\in S} x_i = 1$ or $-1$.

Let $\ell\in \Nbb_+$. For simplicity of notation, we use $\tilde{O}_\eps\pbra{d^m}$ to denote $\pbra{d\cdot \polylog\pbra{n^\ell/\eps}}^m$ for $m,n,d\in\Nbb_+$ and $\eps\in(0,1/2]$. When we omit the subscript $\eps$, it is understood that $\eps=1$. As per this notation, we show a bound of $\tilde{O}\pbra{d^{\ell/2}}$ on the level-$\ell$ Fourier mass of parity decision trees of depth $d$.  We first describe the proof for standard decision trees and then show how to generalize to parity decision trees.
 
\paragraph*{Standard Decision Trees.}  Let $\Tcal$ be a decision tree and for simplicity, assume that every leaf is of depth $d$. Let $v_0,\ldots,v_d$ be a random root-to-leaf path in $\Tcal$ and $\vbm^{(0)},\ldots,\vbm^{(d)}\in\{-1,0,1\}^n$ denote the sequence of partial assignments, i.e., for $j\in[n]$ and $i\in\{0,\ldots,d\}$, let 
\begin{equation} 
\vbm^{(i)}_j=\begin{cases} 
1 & \text{if $x_j$ is fixed to $1$ before reaching $v_i$,} \\ 
-1 & \text{if $x_j$ is fixed to $-1$ before reaching $v_i$,} \\ 
0 &\text{otherwise.}\end{cases} 
\label{def:dt_1}\end{equation}
For $u\in \Rbb^n$, we use $u_S$ to denote $\prod_{j\in S} u_j$. Let $a_S=\sgn\pbra{\widehat{\Tcal}(S)}$ for $|S|=\ell$ and 0 otherwise. Note that 
\begin{equation}\label{eq:dt_2}
\sum_{S:|S|=\ell} \abs{\widehat{\Tcal}(S)} = 
\sum_{S:|S|=\ell} a_S\cdot \widehat{\Tcal}(S) = 
\sum_{S:|S|=\ell} a_S\cdot \E_{v_d}\sbra{ \Tcal(v_d) \vbm^{(d)}_S} 
=\E_{v_d}  \sbra{ \Tcal(v_d)\cdot \pbra{\sum_{S:|S|=\ell} a_S \cdot \vbm^{(d)}_S } }. 
\end{equation}

Thus, to bound $\sum_{S:|S|=\ell} |\widehat{\Tcal}(S)|$ it suffices 
to show that  $\abs{\sum_{S:|S|=\ell}  a_S \cdot \vbm^{(d)}_S}$ is bounded by $\tilde{O}(d^{\ell/2})$ in expectation.
Denote by $X^{(i)}:=\sum_{S:|S|=\ell}  a_S \cdot \vbm^{(i)}_S$ for $i=0,1, \ldots, d$.
We write $X^{(d)}$ as a telescoping sum $X^{(d)} = \sum_{i=1}^{d}\pbra{X^{(i)} - X^{(i-1)}}$. To analyze the difference sequence, observe that in the expression 
$$
X^{(i)} - X^{(i-1)} =
\sum_{S:|S|=\ell} a_S \cdot \pbra{\vbm^{(i)}_S - \vbm^{(i-1)}_S},
$$
if set $S$ contributes to the sum, then $S$ must include the bit queried at the $(i-1)$-th step of the path. Conditioning on $v_0, \ldots, v_{i-1}$, let $x_j$ be the variable queried in $v_{i-1}$, then we have 
 $$
 X^{(i)} - X^{(i-1)} =\sum_{S:|S|=\ell, j \in S} {a_S \cdot \vbm^{(i)}_S} \;=\; x_j \cdot \left( \sum_{S:|S|=\ell, j \in S}a_S \cdot \vbm^{(i-1)}_{S\setminus \{j\}}\right).
 $$
Furthermore, we observe that the sum $\sum_{S:|S|=\ell, j \in S}a_S \cdot \vbm^{(i-1)}_{S\setminus \{j\}}$ is determined by $v_{i-1}$; thus conditioning on $v_0, \ldots, v_{i-1}$ the value of $X^{(i)}-X^{(i-1)}$ is a random coin in $\binpm$ multiplied by some fixed integer. In other words, we get that $X^{(0)}, \ldots, X^{(d)}$ is a martingale with varying step sizes. 

Recall that Azuma's inequality provides concentration bounds for martingales with bounded step sizes, thus now we need to bound $\abs{\sum_{S:|S|=\ell, j \in S}a_S \cdot \vbm^{(i-1)}_{S\setminus \{j\}}}$, which is similar to our initial goal. Put differently, we wish to analyze the sum  
$$
\sum_{S' \subseteq [n]\setminus\{j\}: |S'|=\ell-1} a_{S' \cup \{j\}} \cdot \vbm^{(i-1)}_{S'},
$$ 
which calls for an inductive argument on $\ell$. In addition, since we eventually apply a union bound on all steps, we need to show that $\abs{\sum_{S'} a_{S' \cup \{j\}} \vbm^{(i-1)}_{S'}}$ is bounded with high probability (and not just in expectation).

More generally, to carry an inductive argument we define for any set $T \subseteq [n],|T|\le\ell$ and any $i\in \{0,\ldots, d\}$, the random variable  $$
X_{T}^{(i)}:=\sum_{S\supseteq T :|S|=\ell}  a_S \cdot \vbm^{(i)}_{S\setminus T} = 
\sum_{S' \subseteq \overline{T}: |S'|=\ell-|T|}  a_{S' \cup T} \cdot \vbm^{(i)}_{S'}.
$$
Note that our initial goal was to bound $\abs{X^{(d)}_{\emptyset}} = \abs{X^{(d)}}$, which is analyzed by (reverse) induction on $|T|$ going from larger sets to smaller sets as \Cref{lem:overviewlemma}.

\begin{lemma}\label{lem:overviewlemma}
For all $t\in\{0,\ldots,\ell\}$ and $\eps>0$, the probability that there exist $i\in \{0,\ldots,d\}$ and $T\subseteq[n]$ of size at least $t$ such that $\abs{X_T^{(i)}}\ge \tilde{O}_\eps\pbra{d^{(\ell-t)/2}}$ is at most $\eps\cdot (\ell-t)$.
\end{lemma}
The main observation for the proof is that $X^{(0)}_T,X^{(1)}_T,\ldots,X^{(d)}_T$ is a martingale whose difference sequence consists of terms of the form $X^{(i-1)}_{T'}$ where $T\subsetneq T'$. To see this, if we are querying $x_j$ at $v_{i-1}$, then
\[ 
X^{(i)}_T-X^{(i-1)}_T  =  
\begin{cases}
0 & j\in T,\\
x_j\cdot\pbra{ \underset{  j\notin S\subseteq\overline{T}}{\sum} a_{S\cup T\cup \{j\}}\cdot \vbm_{S}^{(i-1)}  }  = x_j \cdot X _{T\cup j}^{(i-1)} & j\notin T. 
\end{cases}
\]
Note that $X_{T\cup j}^{(i-1)}$ depends only on the history until $v_{i-1}$, and $x_j$ is a uniformly random bit independent of this history, thus $X_T^{(i)}$ is a martingale. The inductive hypothesis implies that with at least $1-\eps\cdot (\ell-t-1)$ probability, $\abs{ X_{T\cup j}^{(i-1)}}\le \tilde{O}_\eps\pbra{d^{(\ell-t-1)/2}}$ for all $T$ of size $t$ and $j\in[n]\setminus T$. Whenever this happens, Azuma's inequality implies that\footnote{Technically this is not true, since a martingale after conditioning may not still be a martingale. We handle this by truncating the martingale when a bad event happens instead of conditioning on the good event.} with probability at least $1-\eps/\pbra{d\cdot n^t}$, we have
\[   
\abs{X_T^{(i)}}\le 2\sqrt{ \log(d\cdot n^t/\eps)} \cdot \sqrt{\sum_{i=1}^d  \tilde{O}_\eps\pbra{d^{\ell-t-1}} } = \tilde{O}_\eps\pbra{d^{(\ell-t)/2}}. 
\]
This, along with a union bound over $T$ of size $t$ and $i\in\{0,\ldots,d\}$ completes the inductive step. The Fourier bound for noisy decision trees can be proved using a similar approach.

\paragraph*{Parity Decision Trees.} 
The basic approach is as before. Let $\Tcal$ be a parity decision tree. As in~\Cref{def:dt_1}, we use $v_i$ and $\vbm^{(i)}$ to denote the random walk and the partial assignments to the variables respectively. We say $v_i$ is \emph{$k$-clean} if
 \begin{equation}\label{eq:pdt_1}
 \forall S\subseteq [n],|S|\le k,\quad  
 \vbm^{(i)}_S=\begin{cases} 1 & \text{if $x_S$ is fixed to $1$ before reaching $v_i$,} \\ 
 -1 & \text{if $x_S$ is fixed to $-1$ before reaching $v_i$,} 
 \\ 0 &\text{otherwise.}
 \end{cases}  
 \end{equation}
For~\Cref{eq:dt_2} to be true, we need that at least $v_d$ is $\ell$-clean.
Note that this is not always true,\footnote{For example, let $S=\{1,2\}$ and consider the parity decision tree whose only query is $x_1x_2$. At any leaf, the value of $x_1x_2$ is fixed, however, the values of $x_1$ and $x_2$ are free, hence $S$ violates \Cref{eq:pdt_1}.} but it is useful as it simplifies the study of high-level Fourier coefficients. To address this issue, we define a \emph{cleanup} process for parity decision trees in which we make additional queries  to ensure that certain key nodes are $k$-clean. We do this by recursively cleaning nodes in a top-down fashion so that for every node $v$ in the original tree $\Tcal$, any node $v'$ in the new tree $\Tcal'$ obtained at the end of the cleanup step for $v$ is $k$-clean. 

The cleanup process is simple to describe: Let $v_1,\ldots,v_d$ be any root-to-leaf path in $\Tcal$. Assume we have completed the cleanup process for $v_1,\ldots,v_{i-1}$. We then query the parity at $v_i$. While there exists a (minimal) set $S$ violating \Cref{eq:pdt_1}, we pick and query an arbitrary coordinate in $S$. Once \Cref{eq:pdt_1} is satisfied, we proceed to the cleanup process for $v_{i+1}$. This process increases the depth by a factor of at most $k$. We set $k=\Theta(\ell\cdot \log(n))$ and work with the new tree $\Tcal'$ of depth $D \le k\cdot d$.

Let $v_0,\ldots,v_{D}$ be a random root-to-leaf path in $\Tcal'$ and $I_i,i\in[D]$ be the set of coordinates fixed due to the query at $v_{i-1}$. Note that this set might be of size larger than $1$.\footnote{For example, suppose we query $x_1 x_2,\, x_1 x_3,\, x_1 x_4$ and finally $x_1$. Then, the last query reveals $4$ coordinates.}
It follows from simple linear algebra that $\sum_{i=1}^{D} \abs{I_i}\le D$. Since $v_D$ is $k$-clean, \Cref{eq:dt_2} holds. Defining $X_T^{(i)}$ exactly as before, our goal is to prove \Cref{lem:overviewlemma} with $D$ instead of $d$. 
The proof is still by induction on $\ell-t$. It turns out that $X_T^{(0)},X_T^{(1)},\ldots,X_T^{(D)}$ is no longer a martingale; instead, $X_T^{(i)}-X_T^{(i-1)}=Y_i+  Z_i$  where 
\begin{equation}\label{eq:pdtdiffseq}
Y_i:=\underset{\substack{\emptyset\neq J\subseteq  I_i \cap  \overline{T} \\ |J|\text{ is even}}}{\sum}  x_J \cdot X_{J\cup T}^{(i-1)}
\quad \text{and} \quad   
Z_i:= \underset{\substack{\emptyset\neq J\subseteq I_i \cap \overline{T} \\ |J|\text{ is odd }}}{\sum} x_J \cdot X_{J\cup T}^{(i-1)}.
\end{equation}
and $Z_i$ (resp., $Y_i$) is an odd (resp., even) polynomial of degree at most $\ell$ over the newly fixed variables $\cbra{ x_j \mid j\in I_i}$. Conditioning on $v_{i-1}$, every pair of random bits $(x_j, x_{j'})$ from $\cbra{x_j \mid j \in I_i}$ is either identical $(x_{j} \equiv x_{j'})$ or opposite $(x_{j} \equiv -x_{j'})$, which means $Y_i$ is a constant and $Z_i$ can be written as $z_i \cdot |Z_i|$ where $|Z_i|$ is a constant and $z_i \sim \binpm$.

For now, let us ignore $Y_i$ and assume that we have a martingale $X_T^{(i)}$ such that $X_T^{(i)}-X_T^{(i-1)}=z_i\cdot \abs{Z_i}$, where $z_i\sim \binpm$ is a uniformly random bit independent of $z_0,\ldots,z_{i-1}$ and $\abs{Z_i}$ depends only on $v_{i-1}$. Combined with an adaptive version of Azuma's inequality, we only need to show the sum of squares of step sizes $\sum_{i=1}^D\abs{Z_i}^2$ is $ \tilde{O}_\eps\pbra{D^{\ell -t}}$ to prove $\abs{X_T^{(i)}}=\tilde{O}_\eps\pbra{D^{(\ell-t)/2}}$. 
By the induction hypothesis, with probability at least $1-\eps\cdot (\ell-t-1)$ the coefficients of $Z_i$ are bounded appropriately. Since $\sum_{i=1}^D|I_i|\le D$ and in particular $|I_i|\le D$, we have
\[ \abs{Z_i}\le \sum_{\text{odd }j\ge1} \binom{\abs{I_i}}{j} \cdot \max_{|T'|=j+t}\abs{X_{T'}^{(i-1)}}   \le \underset{  j\ge 1}{\sum^{\ell-t}} \binom{\abs{I_i}}{j}\cdot \tilde{O}_\eps\pbra{{D}^{(\ell-j-t)/2}} =\tilde{O}_\eps\pbra{ \abs{I_i}\cdot  D^{(\ell-t-1)/2}}  \]
and thus $\sum_{i=1}^D \abs{Z_i}^2 \le D^2\cdot \tilde{O}_\eps\pbra{D^{\ell-t-1}}$. This is too loose for our purpose. 

We instead try to bound the sum of squares of step sizes \emph{with high probability}. 
Imagine for now that $v_{i-1}$ is $2$-clean.\footnote{This assumption immediately implies that $\abs{I_i}\le 1$ and trivially proves our inequality, however, this type of reasoning doesn't generalize to the case when $v_{i-1}$ is not $2$-clean.} Then, the variables $\cbra{x_j\mid j\in I_i}$ are $2$-wise independent conditioning on $v_{i-1}$. This gives
\begin{align*}\begin{split}
\E\sbra{ \abs{Z_i}^2\mid v_{i-1}} &\le  \underset{\text{ odd }  j\ge 1}{\sum} \binom{\abs{I_i}}{j}\cdot \max_{|T'|=j+t}\abs{X_{T'}^{(i-1)}}^2  \le \underset{  j\ge 1}{\sum^{\ell-t}} \binom{\abs{I_i}}{j}\cdot \tilde{O}_\eps\pbra{D^{\ell-j-t}} =\tilde{O}_\eps\pbra{ \abs{I_i}\cdot  D^{\ell-t-1}}
\end{split}\end{align*} 
and thus $\E\sbra{\sum_{i=1}^D \abs{Z_i}^2} \le \tilde{O}_\eps\pbra{D^{\ell-t}}$. 
To show this bound holds with high probability, we use concentration properties of degree-$\ell$ polynomials under $k$-wise independent distributions for $k\gg \ell$. 

In the actual proof, we proceed by conditioning on $C(v_{i-1})$, the nearest ancestor of $v_{i-1}$ that is $k$-clean, instead of conditioning on $v_{i-1}$, which allows to remove the assumption that $v_{i-1}$ is $2$-clean. This is because the queries within a cleanup step are non-adaptive, thus $Z_i$ depends only on $C(v_{i-1})$ and not on $v_{i-1}$. 

Meanwhile, although $X_T^{(i)}$ is not quite a martingale sequence (due to $Y_i$) and the step sizes (i.e., $\abs{Z_i}$) are adaptive and not always bounded, we are nonetheless able to prove an adaptive version of Azuma's inequality of the form $\Pr\sbra{\max_{i\in[D]}\abs{X_T^{(i)}} \ge  \mu+ t\cdot \sigma } \le e^{-\Omega\pbra{t^2}} + \eps$ provided $\Pr\sbra{\pbra{\sum_{i=1}^D\abs{Y_i}\le \mu}\land\pbra{\sum_{i=1}^D\abs{Z_i}^2\le \sigma^2}}\ge1-\eps$. Then it suffices to bound $\sum_{i=1}^D\abs{Y_i}$ similarly to $\sum_{i=1}^D\abs{Z_i}^2$ above.

\subsection{Related Work}
We remark that our proof for level-$\ell$ Fourier growth (even when specialized to the case of standard decision trees) differs from the proofs appearing in \cite{DBLP:conf/focs/Tal20} and \cite{DBLP:journals/eccc/SherstovSW20}.
There, the results were based on decompositions of decision trees. We view our martingale approach as natural and intuitive. We wonder if one can obtain the tight results from \cite{DBLP:journals/eccc/SherstovSW20} using this approach. It seems that the main bottleneck is a union bound on events related to all sets $T\subseteq [n]$ of size at most $\ell$.

Our bounds for level-$1$ improve those obtained by \cite{DBLP:journals/corr/BlaisTW15}. 
They prove that $L_{1,1}(\Tcal) \le O(\sqrt{p\cdot d})$ when $p = \Pr_x[\Tcal(x)=1]$, whereas we obtain a bound of $$L_{1,1}(\Tcal) \le O\pbra{p \sqrt{d}\cdot \log(1/p)}.$$ 
In particular, our bound is almost quadratically better for small values of $p$.
It remains open whether the bound can be further improved to $O\pbra{p \sqrt{d\cdot \log(1/p)}}$, which is the optimal bound for standard decision trees.

We remark that our cleanup technique is inspired by \cite{DBLP:journals/corr/BlaisTW15}, which used cleanup to prove their level-$1$ bound. However, our proof strategies and the way we use the cleanup procedure is quite different than that of \cite{DBLP:journals/corr/BlaisTW15}. 

\paragraph*{Organization.} We make formal definitions in \Cref{sec:prelim}. We state and prove the necessary concentration inequalities in \Cref{sec:useful_ineqs}. We present the cleanup process in \Cref{sec:clean}. We present the Fourier bounds for parity decision trees in \Cref{sec:Fourier_bounds} and for noisy decision trees in \Cref{sec:noisy}.

\section{Preliminaries}\label{sec:prelim}

We use $\log(\cdot)$ to denote the logarithm with base $2$.
We use $[n]$ to denote $\cbra{1,2,\ldots,n}$; and $\binom{[n]}k$ (resp., $\binom{[n]}{\le k}$) to denote the set of all size-$k$ (resp., size-at-most-$k$) sets from $[n]$.
If $S$ is a set from universe $U$, then we write $\overline S$ for $U\setminus S$.
We use $\Ucal_n$ to denote the uniform distribution over $\binpm^n$. We use $\sgn(\mathsf{value})\in\cbra{-1,0,1}$ to denote the sign of $\mathsf{value}$, i.e., 
$\sgn(\mathsf{value})$ equals $-1$ if $\mathsf{value}<0$, $1$ if $\mathsf{value}>0$, and $0$ if $\mathsf{value}=0$.

We use $\Fbb_2=\bin$ to denote the binary field, $\Span\abra{\mathsf{vectors}}$ to denote the subspace spanned by $\mathsf{vectors}$ over $\Fbb_2$.
For a distribution $\Dcal$ we use $x\sim\Dcal$ to represent that $x$ is a random variable sampled from $\Dcal$.
For a finite set $\Xcal$ we use $x\sim\Xcal$ to denote that $x$ is a random variable sampled uniformly from $\Xcal$.
We use the standard notion of $k$-wise independent distribution over $\binpm^n$.
\begin{definition}[$k$-wise independence]
A distribution $\Dcal$ over $\binpm^n$ is $k$-wise independent if for $x\sim \Dcal$ and any $k$-indices $1\le i_1<i_2< \ldots <i_k\le n$, the random variables $(x_{i_1}, \ldots, x_{i_k})$ are uniformly distributed over $\binpm^k$.
\end{definition}

\subsection{Boolean Functions}\label{sec:Fourier}

Here we recall definitions in the analysis of Boolean functions (see \cite{DBLP:books/daglib/0033652} for a detailed introduction).
Let $f\colon\binpm^n\to\Rbb$ be any Boolean function.  
For any $p>0$, the $p$-norm of $f$ is defined as $\vabs{f}_p=\pbra{\E_{x\sim\Ucal_n}\sbra{\abs{f(x)}^p}}^{1/p}$.
For any subset $S\subseteq[n]$, $x_S$ denotes $\prod_{i\in S}x_i$ (in particular, $x_\emptyset=1$).
It is a well-known fact that we can uniquely represent $f$ as a linear combination of $\cbra{x_S}_{S\subseteq[n]}$:
$$
f(x)=\sum_{S\subseteq[n]}\hat f(S)x_S,
$$
where the coefficients $\cbra{\hat f(S)}_{S\subseteq[n]}$ are referred to as the \emph{Fourier coefficients} of $f$ and are given by $\hat f(S)=\E_{x\sim\Ucal_n}\sbra{f(x)x_S}$. 
The above representation expresses $f$ as a multilinear polynomial and is called the Fourier representation of $f$.
We say that $f$ is of degree at most $d$ if its Fourier representation is a polynomial of degree at most $d$, i.e., if $\widehat f(S)=0$ for all $S\subseteq[n],|S|>d$.

\subsection{Parity Decision Trees}

Here we formally define parity decision trees (with Boolean outputs).

\begin{definition}[Parity decision tree]
A \emph{parity decision tree} $\Tcal$ is a representation of a Boolean function $f\colon\binpm^n\to\bin$. 
It consists of a rooted binary tree in which each internal node $v$ is labeled by a non-empty set $Q_v\subseteq[n]$, the outgoing edges of each internal node are labeled by $+1$ and $-1$, and the leaves are labeled by $0$ and $1$. 

On input $x\in\binpm^n$, the tree $\Tcal$ constructs a \emph{computation path} $\Pcal$ from the root to a leaf. Specifically, when $\Pcal$ reaches an internal node $v$ we say that $\Tcal$ \emph{queries} $Q_v$; then $\Pcal$ follows the outgoing edge labeled by $\prod_{i\in Q_v}x_i$. We require that $Q_v$ is not implied by its ancestors' queries.
The output of $\Tcal$ (and hence $f$) on input $x$ is the label of the leaf reached by the computation path.
Conversely, we say $x$ is \emph{consistent with} the path $\Pcal$ if $\Pcal$ is the computation path (possibly ending before reaching a leaf) for $x$.
\end{definition}

We make a few more remarks on a parity decision tree $\Tcal\colon\binpm^n\to\bin$.
\begin{itemize}
\item A node $v$ in $\Tcal$ can be either an internal node or a leaf, and we use $\Tcal(v)\in\bin$ to denote the label on $v$ when $v$ is a leaf.
Meanwhile, we use $\Tcal_v$ to denote the sub parity decision tree starting with node $v$.
\item The \emph{depth} of a node is the number of its ancestors (e.g., the root has depth $0$) and the depth of $\Tcal$ is the maximum depth over all its leaves.
\item 
We say that two parity decision trees $\Tcal$ and $\Tcal'$ are \emph{equivalent} (denoted by $\Tcal\equiv\Tcal'$) if they compute the same function.
\end{itemize}


\subsection{Noisy Decision Trees}

\begin{definition}[Noisy oracle]
A noisy query to a bit $b\in\binpm$ with correlation $\gamma\in [-1,1]$ returns a bit $b'\in\binpm$ where 
\[ 
b'= \begin{cases} 
b & \text{ with probability } (1 + \gamma)/2, \\ 
-b & \text{ with probability } (1- \gamma)/2. \end{cases} 
\]
The cost of a noisy query with correlation $\gamma$ is defined to be $\gamma^2$. 
\end{definition}

\begin{definition}[Noisy decision tree]
A \emph{noisy decision tree $\Tcal$} is a rooted binary tree in which each internal node $v$ is labeled by an index $q_v\in [n]$ and a correlation $\gamma_v\in[-1,1]$. The outgoing edges are labeled by $+1$ and $-1$ and the leaves are labeled by 0 and 1. 

On input $x\in\binpm^n$, the tree $\Tcal$ constructs a \emph{computation path} $\Pcal$ from the root to leaf as follows. When $\Pcal$ reaches an internal node $v$, it makes a noisy query to $x_{q_v}$ with correlation $\gamma_v$ and follows the edge labeled by the outcome of this noisy query. The output of the tree is defined by sampling a root-to-leaf path and returning the label of the leaf. Since the computation path $\Pcal$ is probabilistic, this is an inherently randomized model of computation. We use $\Tcal(x)\in\bin$ to denote the (probabilistic) output of $\Tcal$ on input $x$. We also use $\Tcal(v)\in\bin$ to denote the label on $v$ when $v$ is a leaf. We do \emph{not} require that the indices $q_v$ queried along a path $\Pcal$ are distinct. The \emph{cost} of any path is the sum of costs of the noisy queries along that path; and the cost of $\Tcal$ is the maximum cost of any root-to-leaf path.
\end{definition}

We remark that for any noisy decision tree $\Tcal$, 
its Fourier coefficient $\widehat\Tcal(S)$ is given by $\E\sbra{\Tcal(x)x_S}$ where the expectation is over the randomness of both $x\sim \Ucal_n$ and $\Tcal$.

\section{Useful Concentration Inequalities}\label{sec:useful_ineqs}

We describe useful concentration inequalities in this section.

\subsection{Low Degree Polynomials}\label{sec:low_degree_polys}

We use the fact that low degree polynomials satisfy strong concentration properties under $k$-wise independent distributions.
We will find the following hypercontractive inequality useful.

\begin{theorem}[\cite{Bonami70}, see also
{\cite[$(2,q)$-hypercontractivity]{DBLP:books/daglib/0033652}}]\label{thm:hypercontractivity}
Let $f\colon\binpm^n\to\Rbb$ be a degree-$d$ polynomial. 
Then for any $q\ge 2$, we have $\vabs{f}_q\le(q-1)^{d/2}\vabs{f}_2$.
\end{theorem}

\begin{lemma}\label{lem:low_degree_polys}
Let $f\colon\binpm^n\to\Rbb$ be a degree-$d$ polynomial. 
Let $\Dcal$ be a $2k$-wise independent distribution over $\binpm^n$, where $k\ge d$.
Let $\mu=\E_{x\sim\Dcal}\sbra{f(x)}$ and $\sigma^2=\E_{x\sim\Dcal}\sbra{(f(x)-\mu)^2}$. 
Then for any $\alpha>0$ and any integer $1\le\ell\le k/d$, we have
$$
\E_{x\sim\Dcal}\sbra{\pbra{f(x)-\mu}^{2\ell}}\le\sigma^{2\ell}\cdot\pbra{2\ell-1}^{d\cdot\ell}.
$$
In particular we have
$$
\Pr_{x\sim\Dcal}\sbra{\abs{f(x)-\mu}\ge \alpha\cdot\sigma}
\le \alpha^2\cdot\pbra{\frac{2k}{d\cdot \alpha^{2/d}}}^k.
$$
\end{lemma}
\begin{proof}
Observe that $(f(x)-\mu)^{2\ell}$ is a polynomial of degree at most $2\ell\cdot d\le2k$. Thus its expectation under $\Dcal$ is the same as its expectation under the uniform distribution over $\binpm^n$. By \Cref{thm:hypercontractivity}, we have
$$
\vabs{f-\mu}_{2\ell}\le(2\ell-1)^{d/2}\vabs{f-\mu}_2=\sigma\cdot(2\ell-1)^{d/2}.
$$
Hence by Markov's inequality, we have
\begin{equation*}
\Pr_{x\sim\Dcal}\sbra{\abs{f(x)-\mu}\ge \alpha\cdot\sigma}
\le\frac{\E_{x\sim\Dcal}\sbra{(f(x)-\mu)^{2\ell}}}{(\alpha\cdot\sigma)^{2\ell}}
=\frac{\vabs{f-\mu}_{2\ell}^{2\ell}}{(\alpha\cdot\sigma)^{2\ell}}
\le\frac{(2\ell-1)^{\ell\cdot d}}{\alpha^{2\ell}}.
\end{equation*}

Now we derive the second bound. We only need to focus on the case $\alpha\ge1$ since otherwise the RHS is at least $1$. Then by setting $\ell=\lfloor k/d\rfloor$, we have
\begin{equation*}
\Pr_{x\sim\Dcal}\sbra{\abs{f(x)-\mu}\ge \alpha\cdot\sigma}
\le\frac{(2\lfloor k/d\rfloor-1)^{\lfloor k/d\rfloor\cdot d}}{\alpha^{2\lfloor k/d\rfloor}}
\le\frac{(2k/d)^k}{\alpha^{2(k/d-1)}}
=\alpha^2  \cdot\pbra{\frac{2k}{d\cdot \alpha^{2/d}}}^k.
\tag*{\qedhere}
\end{equation*}
\end{proof}

\subsection{Martingales}\label{sec:martingales}

We show an adaptive version of Azuma's inequality for martingales.
The proof is similar to the inductive proof of the standard Azuma's inequality and thus deferred to \Cref{app:lem:adaptive_azuma_new}. 

\begin{lemma}[Adaptive Azuma's inequality]\label{lem:adaptive_azuma_new}
Let $X^{(0)},\ldots,X^{(D)}$ be a martingale and $\Delta^{(1)},\ldots,\Delta^{(D)}$ be a sequence of magnitudes such that $X^{(0)}=0$ and $X^{(i)}=X^{(i-1)}+\Delta^{(i)}\cdot z^{(i)}$ for $i\in[D]$, where if conditioning on $z^{(1)},\ldots,z^{(i-1)}$,
\begin{itemize}
\item[(1)] $z^{(i)}$ is a mean-zero random variable and $\abs{z^{(i)}}\le1$ always holds;
\item[(2)] $\Delta^{(i)}$ is a fixed value.
\end{itemize}
If there exists some constant $U\ge0$ such that $\sum_{i=1}^D\abs{\Delta^{(i)}}^2\le U$ always holds, then for any $\beta\ge0$ we have
$$
\Pr\sbra{\max_{i=0,1,\ldots,D}\abs{X^{(i)}}\ge\beta\cdot\sqrt{2U}}\le 2\cdot e^{-\beta^2/2}.
$$
\end{lemma}

Next, we generalize \Cref{lem:adaptive_azuma_new} as follows.
\begin{lemma}\label{lem:adaptive_azuma}
Let $m\ge1$ be an integer. 
For each $t\in[m]$, let $X_t^{(0)},\ldots,X_t^{(D)}$ be a sequence of random variables and $\Delta_t^{(1)},\ldots,\Delta_t^{(D)}$ be a sequence of magnitudes such that $X_t^{(0)}=0$ and $X_t^{(i)}=X_t^{(i-1)}+\Delta_t^{(i)}\cdot z_t^{(i)}+\mu_t^{(i)}$ for $i\in[D]$, where if conditioning on $z_t^{(1)},\ldots,z_t^{(i-1)}$,
\begin{itemize}
\item[(1)] $z_t^{(i)}$ is a mean-zero random variable and $\abs{z_t^{(i)}}\le1$ always holds;
\item[(2)] $\Delta_t^{(i)}$ is a fixed value and $\mu_t^{(i)}$ is a random variable.
\end{itemize}
If there exist some constants $U,V\ge0$ and $\eta\in[0,1]$ such that
$$
\Pr\sbra{\exists t\in[m],~\pbra{\sum_{i=1}^D\abs{\Delta_t^{(i)}}^2>U}\lor\pbra{\sum_{i=1}^D\abs{\mu_t^{(i)}}>V}}\le\eta,
$$
then for any $\beta\ge0$ we have
$$
\Pr\sbra{\exists t\in[m],\max_{i=0,1,\ldots,D}\abs{X_t^{(i)}}\ge V+\beta\cdot\sqrt{2U}}\le\eta+2m\cdot e^{-\beta^2/2}.
$$
\end{lemma}
\begin{proof}
We divide the proof into the following two cases.

\paragraph*{\fbox{Case $\eta=0$.}}
Let $\widehat{X}_t^{(i)}=X_t^{(i)}-\sum_{j=1}^i\mu_t^{(j)}$ for each $t$ and $i$. 
Then $\abs{X_t^{(i)}}=\abs{\widehat{X}_t^{(i)}+\sum_{j=1}^i\mu_t^{(j)}}\le V+\abs{\widehat{X}_t^{(i)}}$.
By a union bound, it suffices to show for any fixed $t$, we have 
$$
\Pr\sbra{\max_{i=0,1,\ldots,D}\abs{\widehat{X}_t^{(i)}}\ge \beta\cdot\sqrt{2U}}\le 2\cdot e^{-\beta^2/2},
$$
which follows from \Cref{lem:adaptive_azuma_new}.

\paragraph*{\fbox{Case $\eta\ge0$.}}
Consider $\widetilde{X}_t^{(0)},\ldots,\widetilde{X}_t^{(D)}$ defined by setting $\widetilde{X}_t^{(0)}=0$ and $\widetilde{X}_t^{(i)}=\widetilde{X}_t^{(i-1)}+\widetilde{\Delta}_t^{(i)}\cdot z_t^{(i)}+\widetilde{\mu}_t^{(i)}$, where
$$
\widetilde{\Delta}_t^{(i)}=\begin{cases}
\Delta_t^{(i)} & \sum_{j=1}^i\abs{\Delta_t^{(j)}}^2\le U,\\
0 & \text{otherwise},
\end{cases}
\quad\text{and}\quad
\widetilde{\mu}_t^{(i)}=\begin{cases}
\mu_t^{(i)} & \sum_{j=1}^i\abs{\mu_t^{(j)}}\le V,\\
0 & \text{otherwise}.
\end{cases}
$$
Then Item (1) and (2) hold for $\pbra{\widetilde{X}_t^{(i)}}_{t,i}$ and $\pbra{\widetilde{\Delta}_t^{(i)}}_{t,i},\pbra{\widetilde{\mu}_t^{(i)}}_{t,i}$.

Note that $\Pr\sbra{\exists t\in[m],i\in\cbra{0,1\ldots,D},\widetilde{X}_t^{(i)}\neq X_t^{(i)}}\le\eta$ and $\sum_{i=1}^D\abs{\widetilde{\Delta}_t^{(i)}}^2\le U,\sum_{i=1}^D\abs{\widetilde{\mu}_t^{(i)}}\le V$ always.
Hence from the previous case, we have
\begin{align*}
&\phantom{\le}\Pr\sbra{\exists t\in[m],\max_{i=0,1,\ldots,D}\abs{X_t^{(i)}}\ge V+\beta\cdot\sqrt{2U}}\\
&\le\Pr\sbra{\exists t\in[m],i\in\cbra{0,1\ldots,D},~\widetilde{X}_t^{(i)}\neq X_t^{(i)}}+\Pr\sbra{\exists t\in[m],\max_{i=0,1,\ldots,D}\abs{\widetilde{X}_t^{(i)}}\ge V+\beta\cdot\sqrt{2U}}\\
&\le\eta+2m\cdot e^{-\beta^2/2}.
\tag*{\qedhere}
\end{align*}
\end{proof}

\section{How to Clean Up Parity Decision Trees}\label{sec:clean}

In this section we show how to \emph{clean up} the given parity decision tree to make it easier to analyze. 

\subsection[k-cleanness]{$k$-cleanness}

It will be useful to identify $\Fbb_2^n$ with $\binpm^n$ by $\Enc\colon\pbra{x_1, \ldots, x_n} \mapsto \pbra{(-1)^{x_1}, \ldots, (-1)^{x_n}}$.
For a subset $X \subseteq \Fbb_2^n$ we will denote $\Enc(X) = \{\Enc(x): x\in X\}$.
Thus, we may think of Boolean functions also as $f\colon \Fbb_{2}^n \to \bin$.
We observe that under this representation of the input, a parity decision tree $\Tcal : \Fbb_2^n \to \bin$ indeed queries parity functions (i.e., linear functions over $\Fbb_2$) of the input bits $x \in \Fbb_2^n$ and decides whether to go left or right based on their outcome. Thus, the set of all possible inputs in $\Fbb_2^n$ that reach a given node in a parity decision tree is an affine subspace of $\Fbb_2^n$.

We introduce some notation. 

\begin{notation}
Let $\Tcal\colon\binpm^n\to\bin$ be a parity decision tree and let $v$ be a node in it. 
\begin{itemize}
\item We use $\Pcal_v \subseteq \binpm^n$ to denote the set of all points reaching node $v$. 
Note that $\Pcal_v = \Enc(H_v + a)$ where $H_v$ is a linear subspace of $\Fbb_2^n$ of dimension $n - \depth(v)$ and $a \in \Fbb_2^n$.
\item For any $S\subseteq[n]$, we define $\widehat{\Pcal_v}(S)=\E_{x\sim\Pcal_v}[x_S]$.
\item We use $\Scal_v$ to denote all fully correlated sets with $\Pcal_v$, i.e., $\Scal_v=\cbra{S\subseteq[n]\mid\widehat{\Pcal_v}(S)\in\binpm}$.
We observe that if $\Pcal_v = \Enc(H_v + a)$, then $\Scal_v = H_v^{\perp}$. Additionally, if the queries on the path from root to $v$ are $Q_{v_0}, \ldots, Q_{v_{i-1}}$, then $\Scal_v = \Span\langle{\{Q_{v_0}, \ldots, Q_{v_{i-1}}\}\rangle}$.
\item If $v$ is an internal node, then define $J(v)$ as the set of newly fixed coordinates after querying $Q_v$, i.e., $i\in J(v)$ iff $\cbra{i}\notin\Scal_v$ but $\cbra{i}\in\Span\abra{\Scal_v\cup\cbra{Q_v}}$.
\end{itemize}
\end{notation}

The following simple fact shows that there is no ``somewhat'' correlated set.
\begin{fact}\label{fct:coefficients_of_Pv}
For any parity decision tree $\Tcal$ and any node $v$ in $\Tcal$, $\widehat{\Pcal_v}(S)\in\cbra{+1,0,-1}$ holds for any set $S$.
\end{fact}
\begin{proof}
Since $\Pcal_v = \Enc(H_v+a)$ where $H_v + a$ is an affine subspace, $\Pcal_v$ falls into one of the following $3$ cases: (a) all points in $\Pcal_v$ satisfy $\chi_S(x) = 1$, (b) all points satisfy $\chi_S(x) = -1$, (c) exactly half of the points satisfy $\chi_S(x) = 1$.
\end{proof}

Let $\Scal\subseteq\Fbb_2^n$ be a subspace and $S\subseteq[n]$. For simplicity, we write $S\in\Scal$ iff the indicator vector of $S$ is contained in $\Scal$. 
Now we describe the desired property: \emph{$k$-clean}.

\begin{definition}[$k$-clean subspace and mess-witness]
Let $k$ be a positive integer.
A subspace $\Scal$ is \emph{$k$-clean} if for any set $S\in\Scal$ such that $|S|\le k$, we have that $\cbra{i}\in\Scal$ holds for any $i\in S$. 

Moreover, when $\Scal$ is not $k$-clean, we say $i$ is a \emph{mess-witness} if there exists some $S\ni i,|S|\le k$ such that $S\in\Scal$ but $\cbra{i}\notin\Scal$.
\end{definition}

\begin{definition}[$k$-clean parity decision tree]\label{def:k_clean_PDT}
A parity decision tree $\Tcal$ is \emph{$k$-clean} if the following holds:
\begin{itemize}
\item For any internal node $v$, either (a) $\Scal_v$ is $k$-clean, or (b) $Q_v = \{i\}$ where $i$ is a mess-witness for $\Scal_v$. Moreover, we say $v$ is \emph{$k$-clean} if (a) holds; and we say $v$ is \emph{cleaning} if (b) holds.
\item For any leaf $v$, $\Scal_v$ is $k$-clean (in such a case, we say that $v$ is \emph{$k$-clean}).
\item For any $k$-clean internal node $v$, $\Tcal_v$ starts with $\ell(v)$ non-adaptive queries\footnote{This means for any $i\in\cbra{0,1\ldots,\ell(v)-1}$, all nodes of depth $i$ in $\Tcal_v$ make the same query.} where $\ell(v)\ge1$. 
In addition, for any $i\in\cbra{1,\ldots,\ell(v)-1}$, any node of depth $i$ in $\Tcal_v$ is cleaning; and all node of depth $\ell(v)$ are $k$-clean.\footnote{This ``leveled adaptive'' condition is required just for convenience of proofs. In fact, one can show that the first few queries in $\Tcal_v$ can be rearranged to make sure they are non-adaptive until we reach a $k$-clean node. See \Cref{lem:clean_subspace}.}
\end{itemize}
\end{definition}

\begin{example}\label{exp:k_clean_PDT}
If $\Tcal$ is a decision tree (i.e., $|Q_v|\equiv1$ for any internal node $v$) then it is $k$-clean for any $k$, where each internal node is $k$-clean.

If $\Tcal$ is the depth-$1$ parity decision tree for $\Tcal(x)=x_1x_2x_3$ (i.e., $\Tcal$ only has a root $v_0$ querying $Q_{v_0}=\cbra{1,2,3}$), then it is $2$-clean but not $3$-clean, since for either leaf $v$ we have $\cbra{1,2,3}\in\Scal_v$ but $\cbra{1}\notin\Scal_v$.
\end{example}

The benefit of having a $k$-clean parity decision tree is that it makes the expression of Fourier coefficients simpler.

\begin{lemma}\label{lem:clean_to_Fourier}
Let $\Tcal\colon\binpm^n\to\bin$ be a $k$-clean parity decision tree and let $S$ be a set of size $\ell\le k$.
Let $v_0,\ldots,v_d$ be a random root-to-leaf path.
Define $\vbm^{(0)},\ldots,\vbm^{(d)}\in\cbra{-1,0,+1}^n$ by setting $\vbm^{(i)}_j=\widehat{\Pcal_{v_i}}(j)$ for each $i,j$. Recall that $\vbm^{(d)}_S = \prod_{j\in S} \vbm^{(d)}_j$.
Then we have
$$
\widehat\Tcal(S)
=\E_{v_0,\ldots,v_d}\sbra{\Tcal(v_d)\cdot\vbm^{(d)}_S}.
$$
\end{lemma}
\begin{proof}
Observe that for any $j\in J(v_i)\subseteq J$, the $j$-th coordinate is fixed after querying $Q_{v_i}$. 
Therefore we have
$$
\widehat\Tcal(S)
=\E_{y\sim\Ucal_n}\sbra{\Tcal(y)\cdot y_S}
=\E_{v_0,\ldots,v_d}\sbra{\Tcal(v_d)\cdot\E_{y\sim\Pcal_{v_d}}\sbra{y_S}}
=\E_{v_0,\ldots,v_d}\sbra{\Tcal(v_d)\cdot\widehat{\Pcal}_{v_d}(S)}
$$
By \Cref{fct:coefficients_of_Pv}, $\widehat{\Pcal}_{v_d}(S)\neq0$ iff $S\in\Scal_{v_d}$, which, due to $\ell \le k$ and $v_d$ being a $k$-clean leaf, is equivalent to all coordinates in $S$ being fixed along this path. Hence $\widehat{\Pcal}_{v_d}(S)=\prod_{j\in S}\vbm_j^{(d)}$.
\end{proof}

\subsection{Cleanup Process}

We first analyze the cleanup process for a subspace.\footnote{The $k=2$ case of \Cref{lem:clean_subspace} is essentially \cite[Proposition 3.5]{DBLP:journals/corr/BlaisTW15}. However there is a gap in their proof. For example, if the parity decision tree non-adaptively queries $x_1x_2x_3x_4, x_1x_5, x_2x_6$ in order, then their analysis fails.}

\begin{lemma}[Clean subspace]\label{lem:clean_subspace}
Let $k\ge2$ be an integer and $\Scal$ be a subspace of rank at most $d$. We construct a new subspace $\Scal'$ (initialized as $\Scal$) as follows: while $\Scal'$ is not $k$-clean, we continue to update $\Scal'\gets\Span\abra{\Scal'\cup\cbra{\cbra{i}}}$ with some mess-witness $i$. 
Then $\rank(\Scal')\le d\cdot k$ and any update choice of mess-witnesses will result in the same final subspace $\Scal'$.
\end{lemma}
\begin{proof}
Assume $\Scal$ is a subspace of $\Fbb_2^n$. Then first note that the number of updates is finite, since we can update for at most $n$ times.

Next we show that the number of updates and the final $\Scal'$ does not depend on the choice of mess-witnesses. We do so by an exchange argument.
Let $i_1,\ldots,i_r$ and $i_1',\ldots,i_{r'}'$ be two rounds of execution using different mess-witnesses.
Then there exists some $t<\min\cbra{r,r'}$ such that $i_j=i_j'$ for all $j\le t$, but $i_{t+1}\neq i_{t+1}'$. 
Let $\Scal_t=\Span\abra{\Scal\cup\cbra{\cbra{i_1},\ldots,\cbra{i_t}}}$. Then there exist $S\ni i_{t+1}$ and $S'\ni i_{t+1}'$ (possibly $S=S'$) such that $S,S'\in\Scal_t$ but $\cbra{i_{t+1}},\cbra{i_{t+1}'}\notin\Scal_t$. Since the final subspace is $k$-clean, we know there exists some $T\ge t$ such that
$$
\cbra{i_{t+1}}\notin\Span\abra{\Scal\cup\cbra{\cbra{i_1'},\ldots,\cbra{i_T'}}}
\quad\text{but}\quad
\cbra{i_{t+1}}\in\Span\abra{\Scal\cup\cbra{\cbra{i_1'},\ldots,\cbra{i_{T+1}'}}},
$$
which means $\cbra{i_{T+1}',i_{t+1}}\in\Span\abra{\Scal\cup\cbra{\cbra{i_1'},\ldots,\cbra{i_T'}}}$. Hence we can safely replace $i_{T+1}'$ with $i_{t+1}$, and then swap $i_{t+1}$ with $i_{t+1}'$. We can perform this process as long as $(i_1,\ldots,i_r)\neq(i_1',\ldots,i_{r'}')$, which means $r=r'$ and the final $\Scal'$ is always the same.

For any subspace $\Hcal$, we define $\rank_1(\Hcal)=\abs{\cbra{i\mid\cbra{i}\in\Hcal}}$ and thus $\rank(\Hcal)-\rank_1(\Hcal)\ge0$.
Now we analyze the following particular way to construct $\Scal'$: We initialize $\Scal'$ as $\Scal$. While $\Scal'$ is not $k$-clean, we find a minimal $S=\cbra{i_1,\ldots,i_s}\in\Scal'$ such that $i_1$ is a mess-witness; then we update $\Scal'\gets\Span\abra{\Scal'\cup\cbra{\cbra{i_1},\ldots,\cbra{i_{s-1}}}}$.
Note that before the update, $1<s\le k$ and $\cbra{i_j}\notin\Scal'$ holds for each $j\in[s]$, since $S$ is minimal and $\Scal'$ is not $k$-clean. Thus after the update, $\rank(\Scal')$ grows by $s-1\le k-1$ and $\rank_1(\Scal')$ grows by $s$, which means $\rank(\Scal')-\rank_1(\Scal')$ shrinks by $1$. Hence we have at most $\rank(\Scal)-\rank_1(\Scal)\le d$ updates before $\Scal'$ is $k$-clean; and the final $\Scal'$ has rank at most $\rank(\Scal)+(k-1)\cdot d\le d\cdot k$.
\end{proof}

We now show how to convert an arbitrary parity decision tree into a $k$-clean parity decision tree which still has a small depth and fixes a small number of variables along each path. The latter quantity is in fact bounded by the depth as shown in \Cref{fct:depth_to_singleton}.

\begin{fact}\label{fct:depth_to_singleton}
Let $\Tcal$ be a depth-$d$ parity decision tree.
Let $v_0,\ldots,v_{d'}$ be any root-to-leaf path.
Then we have $\sum_{i=0}^{d'-1}|J(v_i)|\le d'$.
\end{fact}
\begin{proof}
Observe that $\sum_{i=0}^{d'-1}|J(v_i)|=\abs{\cbra{i\mid \cbra{i}\in\Span\abra{Q_{v_0},\ldots,Q_{v_{d'-1}}}}}\le d'$.
\end{proof}

\begin{corollary}\label{cor:depth_to_singleton}
Let $\Tcal$ be a depth-$D$ $k$-clean parity decision tree.
Let $v_0,\ldots,v_{D'}$ be any root-to-leaf path where at most $d$ of the nodes $v_0, \ldots, v_{D'-1}$ are $k$-clean.
Then $\sum_{i:|J(v_{i-1})|>1} \abs{J(v_i)} \le 2d$. 
\end{corollary}
\begin{proof}
By \Cref{fct:depth_to_singleton} we have $\sum_{i=0}^{D'-1}{|J(v_i)|-1} \le 0$.
Since any $v_i$ with $J(v_i)=\emptyset$ is not cleaning and therefore must be $k$-clean.
Thus
$$
\sum_{i: |J(v_i)|>1} |J(v_i)|-1 \le \abs{\{i: J(v_i)=\emptyset\}}
\le d.
$$
For $|J(v_i)|>1$, we have $|J(v_i)|-1\ge|J(v_i)|/2$ and thus
$\sum_{i: |J(v_i)|>1} |J(v_i)| \le 2d$.
\end{proof}

\begin{lemma}[Clean parity decision tree]\label{lem:clean_PDT}
Let $k\ge2$ be an integer. 
Let $\Tcal$ be an arbitrary depth-$d$ parity decision tree. Then there exists a $k$-clean parity decision tree $\Tcal'$ of depth at most $d\cdot k$ equivalent to $\Tcal$. Moreover, any root-to-leaf path in $\Tcal'$ has at most $d$ nodes that are $k$-clean.
\end{lemma}
\begin{proof}
We build $\Tcal'$ by the following recursive algorithm. An example of the algorithm is provided in \Cref{fig:clean_PDT}

\begin{algorithm}[ht]
\caption{Clean parity decision tree: build $\Tcal'$ from $\Tcal$}\label{alg:clean_PDT}
\DontPrintSemicolon
\LinesNumbered
\SetKwProg{proc}{Procedure}{}{}
\SetKwFunction{Build}{Build}
\KwIn{an arbitrary depth-$d$ parity decision tree $\Tcal$}
\KwOut{a parity decision tree $\Tcal'$ with desired properties}
$r\gets$ root of $\Tcal$\;
Initialize the root of $\Tcal'$ as $r'$\;
\Build{$r,r',1$}\;
\proc{\Build{$v,v',\ell$}}{
\tcc{$(v,v')$ are the current nodes on $(\Tcal,\Tcal')$; $\ell$ is the recursion depth.}
\lIf{$v$ is a leaf}{Label ${v'}$ with the label of $v$}
\Else{
$(v_-,v_+)\gets$ the left and right child of $v$\; 
\lIf{$\widehat{\Pcal_{v'}}(Q_v)=-1$}{\Build{$v_-,v',\ell+1$}}
\lElseIf{$\widehat{\Pcal_{v'}}(Q_v)=+1$}{\Build{$v_+,v',\ell+1$}}
\Else(\tcc*[f]{$\widehat{\Pcal_{v'}}(Q_v)=0$ due to \Cref{fct:coefficients_of_Pv}}){
$Q_{v'}\gets Q_v$\;
$(v'_-,v'_+)\gets$ the left and right child of $v'$\;
Initialize $O\gets\emptyset$\;
\While{$\Span\abra{\Scal_{v'}\cup\cbra{Q_{v'}}\cup O}$ is not $k$-clean}{
Update $O\gets O\cup\cbra{\cbra{i}}$, where $i$ is a \emph{mess-witness}
}
$\Tcal'$ non-adaptively queries every set (which is a singleton) in $O$ under $v'$ in arbitrary order\;
\lForEach{leaf $\hat v$ under $v'_-$}{\Build{$v_-,\hat{v},\ell+1$}\;
\lForEach{leaf $\hat v$ under $v'_+$}{\Build{$v_+,\hat{v},\ell+1$}}
}
}
}
}
\end{algorithm}

\begin{figure}[ht]
\centering
\scalebox{1.3}{\begin{tikzpicture}[
emptyC/.style={draw,circle,inner sep=2pt}, 
RedC/.style={circle,inner sep=2pt,fill=red},
BlueC/.style={circle,inner sep=2pt,fill=blue},
outE/.style={->,>=stealth,thick},
buildE/.style={-,dashed,very thin},
node distance=1cm]

\clip (-2,-3.8) rectangle (7,0.5);

\begin{scope}
\node[emptyC, label={above:$x_1x_2$}] (v0) at (0,0) {};
\node[emptyC, label={[xshift=-1mm]above:$x_2$}] (v1) [below left=of v0] {};
\node[emptyC, label={[xshift=1mm]above:$x_4$}] (v2) [below right=of v0] {};
\node[emptyC, label={[xshift=-2mm]above:$x_3$}] (v3) [below left=of v1,xshift=5mm] {};

\node (q1) [below right=of v1,xshift=-5mm] {\!$1$};
\node (q2) [below left=of v2,xshift=5mm] {$1$\!};
\node (q3) [below right=of v2,xshift=-5mm] {\!$0$};
\node (q4) [below left=of v3,xshift=7mm] {$0$\!};
\node (q5) [below right=of v3,xshift=-7mm] {\!$1$};

\foreach \startnode/\endnode in {v0/v1, v0/v2, v1/v3, v1/q1, v2/q2, v2/q3, v3/q4, v3/q5} 
\draw[outE] (\startnode) -- (\endnode);
\end{scope}
    
\begin{scope}[xshift=5cm]
\node[RedC, label={above:$x_1x_2$}] (u0) at (0,0) {};
\node[BlueC, label={[xshift=-1mm]above:$x_1$}] (u1) [below left=of u0] {};
\node[BlueC, label={[xshift=1mm]above:$x_1$}] (u2) [below right=of u0] {};
\node[RedC, label={[xshift=2mm]above:$x_3$}] (u3) [below right=of u1,xshift=-5mm] {};
\node[RedC, label={[xshift=-2mm]above:$x_4$}] (u4) [below left=of u2,xshift=5mm] {};
\node[RedC, label={[xshift=2mm]above:$x_4$}] (u5) [below right=of u2,xshift=-5mm] {};

\node (r1) [below left=of u1,xshift=5mm] {$1$\!};
\node (r2) [below left=of u3,xshift=7mm] {$0$\!};
\node (r3) [below right=of u3,xshift=-7mm] {\!$1$};
\node (r4) [below left=of u4,xshift=7mm] {$1$\!};
\node (r5) [below right=of u4,xshift=-7mm] {\!$0$};
\node (r6) [below left=of u5,xshift=7mm] {$1$\!};
\node (r7) [below right=of u5,xshift=-7mm] {\!$0$};

\foreach \startnode/\endnode in {u0/u1, u0/u2, u1/u3, u2/u4, u2/u5, u1/r1, u3/r2, u3/r3, u4/r4, u4/r5, u5/r6, u5/r7}
\draw[outE] (\startnode) -- (\endnode);
\end{scope}

\draw[buildE] (v0) -- (u0);
\draw[buildE] (v1) -- (r1);
\draw[buildE] (v2) edge [bend left=28] (u4);
\draw[buildE] (v2) edge [bend left=40] (u5);
\draw[buildE] (v3) edge [bend left=25] (u3);
\draw[buildE] (q1) edge [bend left=30] (r1);
\draw[buildE] (q2) edge [out=-25, in=-130] (r4);
\draw[buildE] (q2) edge [out=-15, in=-110] (r6);
\draw[buildE] (q3) edge [out=-20, in=140] (r5);
\draw[buildE] (q3) edge [out=-20, in=140] (r7);
\draw[buildE] (q4) edge [out=-25, in=-160] (r2);
\draw[buildE] (q5) edge [out=-25, in=150] (r3);
\end{tikzpicture}}
\caption{An example of the cleanup process with $k=2$ where the LHS is $\Tcal$ and the RHS is $\Tcal'$. All the left (resp., right) outgoing edges are labeled with $-1$ (resp., $+1$). {\color{red} Red nodes} and leaves are $k$-clean, and {\color{blue} blue nodes} are cleaning (i.e., non-adaptive queries). Nodes connected with dashed curves are invoked by \textsf{Build}.}\label{fig:clean_PDT}
\end{figure}
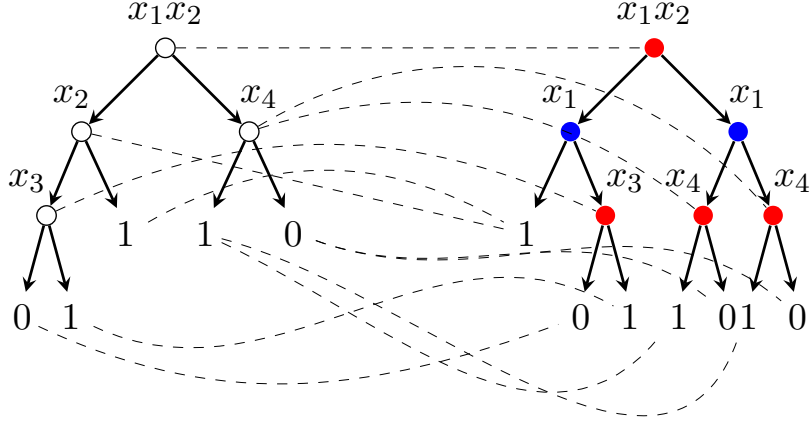

We now prove the correctness of \Cref{alg:clean_PDT}, which is guaranteed by the following claims.
\begin{itemize}
\item \fbox{For any internal node $v'\in\Tcal'$, $Q_{v'}$ is not implied by its ancestors' queries.} 
By \Cref{fct:coefficients_of_Pv}, this is equivalent to $Q_{v'}\notin\Scal_{v'}$, which follows from the conditions in \textsf{Line 8/9/13}.
\item \fbox{The depth of $\Tcal'$ is at most $d\cdot k$.}
Let $v_0,\ldots,v_{d'}$ be any root-to-leaf path of $\Tcal$ and let $\Pcal'$ be its corresponding path in $\Tcal'$. Then the construction process of $\Pcal'$ corresponds to the cleanup process for $\Span\abra{Q_{v_0},\ldots,Q_{v_{d'-1}}}$ in \Cref{lem:clean_subspace}; hence the depth of $\Tcal'$ equals $\rank(\Scal')\le d'\cdot k\le d\cdot k$ where $\Scal'$ is the $k$-clean subspace produced by applying \Cref{lem:clean_subspace}.
\item \fbox{$\Tcal\equiv\Tcal'$ and any root-to-leaf path in $\Tcal'$ has at most $d$ $k$-clean nodes.}
This is evident from the algorithm, as $\Tcal'$ only refines $\Tcal$ by inserting cleaning nodes.
\item \fbox{Whenever we call \Build{$\cdot,v',\cdot$}, $v'$ is $k$-clean.}
We prove by induction on $\ell$. The base case \textsf{Line 3} is obvious. For \textsf{Line 8/9}, we recurse on the same $v'$, which is $k$-clean by induction. 
For \textsf{Line 17/18}, note that $\Scal_{\hat{v}}=\Span\abra{\Scal_{v'}\cup\cbra{Q_{v'}}\cup O}$; hence from the condition in \textsf{Line 13}, it is $k$-clean.
\item \fbox{Nodes created in \textsf{Line 16} are cleaning.}
Let $o=|O|$ and let $i_1,i_2,\ldots,i_o$ be the query order. 
For any $j\in[o]$, let $v'_j$ be any one of the nodes created for $i_j$, then 
$$\Scal_{v'_j}=\Span\abra{\Scal_{v'}\cup\cbra{Q_{v'}}\cup\cbra{\cbra{i_1},\ldots,\cbra{i_{j-1}}}},$$
which is not $k$-clean by \textsf{Line 13}; hence $v'_j$ is cleaning by the condition in \textsf{Line 13}.
\qedhere
\end{itemize}
\end{proof}

\section{Fourier Bounds for Parity Decision Trees}\label{sec:Fourier_bounds}

Our goal in this section is to prove \Cref{thm:informal} with detailed bounds provided.

\subsection[Level-1 Bound]{Level-$1$ Bound}\label{sec:level_1}

We first prove the concentration result for level-$1$.
We start with the following simple bound for general parity decision trees.
\begin{lemma}\label{lem:level_1_D}
Let $\Tcal\colon\binpm^n\to\bin$ be a depth-$D$ parity decision tree.
Let $v_0,\ldots,v_{D'}$ be any root-to-leaf path.
Define $\vbm^{(0)},\ldots,\vbm^{(D')}\in\cbra{-1,0,+1}^n$ by setting $\vbm^{(i)}_j=\widehat{\Pcal_{v_i}}(j)$ for each $0\le i\le D'$ and $j\in[n]$.
Then for any $a_1,\ldots,a_n\in\cbra{-1,0,1}$, we have $\abs{\sum_{j=1}^na_j\cdot\vbm_j^{(D')}}\le D' \le D$.
\end{lemma}
\begin{proof}
Note that the set of non-zero coordinates in $\vbm^{(D')}$ is exactly $\bigcup_{i=0}^{D'-1}J(v_i)$. Hence by \Cref{fct:depth_to_singleton}, we have
\begin{equation*}
\abs{\sum_{j=1}^na_j\cdot\vbm_j^{(D')}}
\le\sum_{j=1}^n\abs{\vbm_j^{(D')}}
=\sum_{i=0}^{D'-1}|J(v_i)|
\le D' \le D.
\tag*{\qedhere}
\end{equation*}
\end{proof}

Now we give an improved bound for $k$-clean parity decision trees.
To do so, we need one more notation which will be crucial in our analysis.
\begin{notation}
Let $\Tcal$ be a $k$-clean parity decision tree.
For any node $v$, we define $C(v)$ as the nearest ancestor of $v$ (including itself) that is $k$-clean.
\end{notation}
\begin{lemma}\label{lem:level_1_sqrtD}
There exists a universal constant $\kappa\ge1$ such that the following holds.
Let $\Tcal\colon\binpm^n\!\to\bin$ be a depth-$D$ $2k$-clean parity decision tree where $k\ge1$ and any root-to-leaf path has at most $d$ nodes that are $2k$-clean.

Let $v_0,\ldots,v_{D'}$ be a random root-to-leaf path.
Define $\vbm^{(0)},\ldots,\vbm^{(D')}\in\cbra{-1,0,+1}^n$ by setting $\vbm^{(i)}_j=\widehat{\Pcal_{v_i}}(j)$ for each $0\le i\le D'$ and $j\in[n]$.
Then for any $a_1,\ldots,a_n\in\cbra{-1,0,1}$ and any $\eps\le1/2$, we have
$\Pr\sbra{\abs{\sum_{j=1}^na_j\cdot\vbm_j^{(D')}}\ge R(D,d,k,\eps)}\le\eps,$ where
$$
R(D,d,k,\eps)=\kappa\cdot\sqrt{\pbra{D+dk\pbra{\frac{1}\eps}^{\frac1k}}\log\pbra{\frac1\eps}}.
$$
\end{lemma}

In the proof of \Cref{lem:level_1_sqrtD} we will use the following simple claim.
\begin{fact}\label{fct:CS}
Let $p_1, \ldots, p_n$ be a sub-probability distribution, i.e., $p_i\ge 0$ and $\sum_{i=1}^{n} p_i \le 1$.
Let $a_1, \ldots, a_n \in \mathbb{R}$.
Then for any $k\in\Nbb$, we have $\sum_{i=1}^{n} p_i a_i^{2k} \ge \pbra{\sum_{i=1}^{n} p_i a_i^2}^k$.
\end{fact}
\begin{proof}
We add $p_{n+1} = 1-\pbra{\sum_{i=1}^{n}p_i}$ and $a_{n+1} = 0$ so $p$ is a probability distribution. Then the claim follows from $\E[X^k] \ge \E[X]^k$, where random variable $X$ gets value $a_i^2$ with probability $p_i$. 
\end{proof}

\begin{proof}[Proof of~\Cref{lem:level_1_sqrtD}]
Extend $\vbm^{(D'+1)}=\cdots=\vbm^{(D)}$ to equal $\vbm^{(D')}$.
For each $0\le i\le D$, let $X^{(i)}=\sum_{j=1}^na_j\cdot\vbm_j^{(i)}$. 
We define $\delta^{(i)}=0$ for $D'<i\le D$. 
For $1\le i\le D'$, we let
$$
\delta^{(i)}=X^{(i)}-X^{(i-1)}
=\sum_{j=1}^na_j\cdot\pbra{\vbm_j^{(i)}-\vbm_j^{(i-1)}}
=\sum_{j\in J(v_{i-1})}a_j\cdot\vbm_j^{(i)},
$$
where $J(v_{i-1})$ depends only on $C(v_{i-1})$ since $\Tcal_{C(v_{i-1})}$ performs non-adaptive queries before (and possibly even after) reaching  $v_i$.
Note that for the two possible outcomes of querying $Q_{v_i}$, $\vbm_j^{(i)}$ is fixed to $\pm1$ respectively for each $j\in J(v_{i-1})$.
Thus $\delta^{(i)}=\Delta^{(i)}\cdot z^{(i)}$ where $\Delta^{(i)}$ is a fixed value given $z^{(1)},\ldots,z^{(i-1)}$ and $z^{(1)},\ldots,z^{(D')}$ are independent unbiased coins in $\binpm$. 

Since $C(v_{i-1})$ is $2k$-clean, the collection of random variables $\cbra{\vbm_j^{(i)}\mid j\in J(v_{i-1})}$ is $2k$-wise independent conditioning on $C(v_{i-1})$.
Note that $\delta_i$ is a linear function and
$$
\E\sbra{\delta^{(i)}\mid C(v_{i-1})}=0
\quad\text{and}\quad
\E\sbra{\pbra{\delta^{(i)}}^2\mid C(v_{i-1})}=\sum_{j\in J(v_{i-1})}a_j^2\le\abs{J(v_{i-1})}.
$$
By the first bound in \Cref{lem:low_degree_polys}, we have
\begin{equation}\label{eq:level_1_sqrtD_1}
\E\sbra{\pbra{\delta^{(i)}}^{2k}\mid C(v_{i-1})} 
\le(2k-1)^k \cdot\abs{J(v_{i-1})}^k.
\end{equation}
Meanwhile, $\abs{\delta^{(i)}}\le\abs{J(v_{i-1})}$ always.
Our first goal is to bound $\Pr\sbra{\sum_{i=1}^D\pbra{\delta^{(i)}}^2>D+2\alpha^2d}$.
Observe that whenever the event $\sum_{i=1}^D\pbra{\delta^{(i)}}^2>D+2\alpha^2d$ happens, it must be the case that $\sum_{i: |J(v_{i-1})|>1} \pbra{\delta^{(i)}}^2  > 2\alpha^2 d$.
Thus,
\begin{align*}\Pr\sbra{\sum_{i=1}^D\pbra{\delta^{(i)}}^2>D+2\alpha^2d} 
&\le 
\Pr\sbra{\sum_{i: |J(v_{i-1})|>1} \pbra{\delta^{(i)}}^2  > 2\alpha^2 d}\\
&=
\Pr\sbra{\sum_{i: |J(v_{i-1})|>1} \frac{|J(v_{i-1})|}{2d}\cdot \frac{\pbra{\delta^{(i)}}^2}{|J(v_{i-1})|}  > \alpha^2}\\
&\le 
\Pr\sbra{\sum_{i: |J(v_{i-1})|>1} \frac{|J(v_{i-1})|}{2d}\cdot \frac{\pbra{\delta^{(i)}}^{2k}}{|J(v_{i-1})|^k}  > \alpha^{2k}}
\tag{by \Cref{fct:CS} and \Cref{cor:depth_to_singleton}}\\
&=\Pr\sbra{\sum_{i: |J(v_{i-1})|>1} \frac{\pbra{\delta^{(i)}}^{2k}}{|J(v_{i-1})|^{k-1}}  > 2d \cdot \alpha^{2k}}\\
&\le \E\sbra{\sum_{i: |J(v_{i-1})|>1} \frac{\pbra{\delta^{(i)}}^{2k}}{|J(v_{i-1})|^{k-1}}} \cdot \frac{1}{2d \cdot \alpha^{2k}}.
\tag{by Markov's inequality}
\end{align*}
On the other hand,
\begin{align*}
\E\sbra{\sum_{i: |J(v_{i-1})|>1} \frac{\pbra{\delta^{(i)}}^{2k}}{|J(v_{i-1})|^{k-1}}}
&= \sum_{i=1}^{D} \E_{C(v_{i-1})}\sbra{\frac{\indicator_{|J(v_{i-1})|>1}}{|J(v_{i-1})|^{k-1}} \cdot \E\sbra{\pbra{\delta^{(i)}}^{2k}\mid C(v_{i-1})}}\\
&\le \sum_{i=1}^{D} \E_{C(v_{i-1})}\sbra{\indicator_{|J(v_{i-1})|>1}\cdot(2k-1)^k \cdot |J(v_{i-1})|}
\tag{by \Cref{eq:level_1_sqrtD_1}}\\
&=  (2k-1)^k  \cdot \E\sbra{\sum_{i:|J(v_{i-1}|>1} |J(v_{i-1})|}\\
&\le (2k-1)^k \cdot 2d.
\tag{by \Cref{cor:depth_to_singleton}}
\end{align*}
Overall, we have
$$
\Pr\sbra{\sum_{i=1}^D\pbra{\delta^{(i)}}^2>D+2\alpha^2d} \le \frac{(2k-1)^k}{\alpha^{2k}}.
$$

Then by \Cref{lem:adaptive_azuma} with $m=1$, we have 
$$
\Pr\sbra{\abs{X^{(D)}}=\abs{\sum_{j=1}^na_j\cdot\vbm_j^{(D)}}\ge\beta\sqrt{2\cdot\pbra{D+2\alpha^2d}}}
\le2\cdot e^{-\beta^2/2}+\frac{(2k-1)^k}{\alpha^{2k}}.
$$
The desired bound follows from setting 
\begin{equation*}
\alpha=\pbra{\frac{2}\eps}^{\frac1{2k}}\sqrt{2k-1},
\quad\text{and}\quad
\beta=\Theta\pbra{\sqrt{\log\pbra{\frac1\eps}}}.
\tag*{\qedhere}
\end{equation*}
\end{proof}

Now we prove the complete level-$1$ bound for parity decision trees.
\begin{theorem}\label{thm:level_1}
Let $\Tcal\colon\binpm^n\to\bin$ be a depth-$d$ parity decision tree. Let $p=\Pr\sbra{\Tcal(x)=1}\in\sbra{2^{-d},1/2}$.\footnote{If $p<2^{-d}$, then $p=0$ and $\Tcal\equiv0$. If $p>1/2$, we can consider $\tilde\Tcal=1-\Tcal$ by symmetry.}  Then we have 
$$
\sum_{j=1}^n\abs{\widehat\Tcal(j)}\le  
p\cdot \min\cbra{d,
O\pbra{\sqrt{d}\cdot \log\pbra{\frac1p}}}
=O\pbra{\sqrt {d}}.
$$
\end{theorem}
\begin{proof}
For any $i\in[n]$, let $a_i=\sgn\pbra{\widehat\Tcal(i)}$. 
Now we prove the two bounds separately.

\paragraph*{\fbox{First Bound.}}
Let $v_0,\ldots,v_{d'}$ be a random root-to-leaf path in $\Tcal$.
Define $\vbm^{(0)},\ldots,\vbm^{(d')}\in\cbra{-1,0,+1}^n$ by setting $\vbm^{(i)}_j=\widehat{\Pcal_{v_i}}(j)$ for each $0\le i\le d'$ and $j\in[n]$.
Since $\Tcal$ is $1$-clean in itself, by \Cref{lem:clean_to_Fourier} we have
\begin{equation}\label{eq:level_1_1}
\sum_{j=1}^n\abs{\widehat\Tcal(j)}
=\sum_{j=1}^na_i\cdot\widehat\Tcal(j)
=\E_{v_0,\ldots,v_{d'}}\sbra{\Tcal(v_{d'})\cdot\sum_{j=1}^na_j\cdot\vbm_j^{(d')}}
\le\E_{v_0,\ldots,v_{d'}}\sbra{\Tcal(v_{d'})\cdot|V|},
\end{equation}
where $V=\sum_{j=1}^na_j\cdot\vbm_j^{(d')}$.
Hence by \Cref{lem:level_1_D}, we have $\Cref{eq:level_1_1}\le d\cdot\E\sbra{\Tcal(v_{d'})}=p\cdot d$.

\paragraph*{\fbox{Second Bound.}}
By \Cref{lem:clean_PDT}, we construct a $2k$-clean parity decision tree $\Tcal'$ of depth $D\le2d\cdot k$ equivalent to $\Tcal$, where $k=\Theta(\log(1/p))$. Let $U=\sum_{j=1}^na_j\cdot\ubm_j^{(D')}$. Then we have 
\begin{equation}\label{eq:level_1_3}
\sum_{j=1}^n\abs{\widehat\Tcal(j)}
=\sum_{j=1}^n\abs{\widehat{\Tcal'}(j)}
=\E_{u_0,\ldots,u_{D'}}\sbra{\Tcal'(u_{D'})\cdot\sum_{j=1}^na_j\cdot\ubm_j^{(D')}}
\le\E_{u_0,\ldots,u_{D'}}\sbra{\Tcal'(u_{D'})\cdot|U|}.  
\end{equation}
\Cref{lem:level_1_sqrtD} implies that for all $\eps>0$, $\Pr\sbra{|U|\ge R(\eps)}\le\eps$ where 
\[ 
R(\eps) = R(D,d,k,\eps)  = O\pbra{\sqrt{dk\cdot\pbra{\frac{1}\eps}^{\frac1k} \cdot \log\pbra{\frac 1\eps}}}.
\]
For integer $i\ge1$, let $I_{i}=\sbra{R\pbra{p/2^i},R\pbra{p/2^{i+1}}}$ and $I_0=\sbra{0,R(p/2)}$ be intervals. Then for each $i\ge1$, $\Pr\sbra{\abs{U}\in I_i}\le p/2^i$. We also know that $\E_{u_0,\ldots,u_{D'}}\sbra{\Tcal'(u_{D'})}\le p$. Thus,
\begin{align*}\Cref{eq:level_1_3}
&=\E_{u_0,\ldots,u_{D'}}\sbra{\Tcal'(u_{D'})\cdot|U|\cdot \sum_{i=0}^{+\infty} \indicator_{\abs{U}\in I_i}}\\
&\le R\pbra{\frac{p}{2}}\cdot  \E_{u_0,\ldots,u_{D'}}\sbra{\Tcal'(u_{D'})}+ \sum_{i=1}^{+\infty} R\pbra{\frac{p}{2^{i+1}}}\cdot  \E_{u_0,\ldots,u_{D'}}\sbra{\indicator_{\abs{U}\in I_i}}\\
&\le\sum_{i=0}^{+\infty} R\pbra{\frac{p}{2^{i+1}}}\cdot \frac{p}{2^i}\\
&= \sum_{i=0}^{+\infty} O\pbra{ p\cdot \sqrt{ dk \cdot \pbra{\frac{2^{i+1}}{p}}^{\frac1k} \cdot \pbra{\log\pbra{\frac1p}+i+1}}} \cdot \frac{1}{2^i} \\
&= O\pbra{p\cdot\sqrt{dk\cdot \log\pbra{\frac 1p}}} = O\pbra{ p \cdot \sqrt{ d} \cdot \log\pbra{\frac1p}}.\qedhere
\end{align*}
\end{proof}

\subsection[Level-l Bound]{Level-$\ell$ Bound}\label{sec:level_l}

Now we turn to the general levels. 

\begin{lemma}\label{lem:level_l}
There exists a universal constant $\tau\ge1$ such that the following holds.
Let $\ell\ge1$ be an integer.
Let $\Tcal\colon\binpm^n\to\bin$ be a depth-$D$ $2k$-clean parity decision tree where $k\ge4\cdot\ell$ and $n\ge\max\cbra{\tau,k,D}$ and any root-to-leaf path has at most $d$ nodes that are $2k$-clean.

Let $v_0,\ldots,v_{D'}$ be a random root-to-leaf path.
Define $\vbm^{(0)},\ldots,\vbm^{(D')}\in\cbra{-1,0,+1}^n$ by setting $\vbm^{(i)}_j=\widehat{\Pcal_{v_i}}(j)$ for each $0\le i\le D'$ and $j\in[n]$.
Extend $\vbm^{(D'+1)}=\cdots=\vbm^{(D)}$ to equal $\vbm^{(D')}$.
Then for any sequence $a_S\in\cbra{-1,0,1},S\in\binom{[n]}\ell$, any $\eps\le1/2$ and $t\in\{0,\ldots,\ell\}$, we have
$$
\Pr\sbra{\exists t'\in\{0,\ldots,t\},\exists T\in\binom{[n]}{\ell-t'},\exists i\in[D],~
\abs{\sum_{S\subseteq\overline T,|S|=t'}a_{S\cup T}\cdot\vbm_S^{(i)}}\ge M(D,d,k,\ell,t',\eps)}
\le\eps\cdot t,
$$
where we recall that $\vbm_S^{(i)} =\prod_{j\in S} \vbm^{(i)}_j$ and where
$$
M(D,d,k,\ell,t',\eps)=\pbra{\tau\cdot (D+dk)\cdot\pbra{\frac{n^\ell}\eps}^{\frac6k}\log\pbra{\frac{n^\ell}\eps}}^{t'/2}.
$$ 
\end{lemma}
\begin{proof}
We prove the bound by induction on $t=0,1,\ldots,\ell$ and show $\tau=10^4$ suffices.
The base case $t=0$ is trivial, since for any fixed $T$ and $i$, we always have $\abs{a_T\cdot\vbm_\emptyset^{(i)}}\le1=M(D,d,k,\ell,0,\eps)$.

Now we focus on the case where $1\le t\le\ell$. 
For each $0\le i\le D$ and $T\in\binom{[n]}{\ell-t}$, let $$X_T^{(i)}=\sum_{S\subseteq\overline T,|S|=t}a_{S\cup T}\cdot\vbm_S^{(i)}.$$
For $1\le i\le D'$, we have
\begin{align*}
X_T^{(i)}-X_T^{(i-1)}
&=\sum_{S\subseteq\overline T,|S|=t,S\cap J(v_{i-1})\neq\emptyset}a_{S\cup T}\cdot\vbm_S^{(i)}\\
&=\sum_{r=1}^t\sum_{\substack{U\subseteq J(v_{i-1})\cap\overline{T},\\|U|=r}}\vbm_U^{(i)}\sum_{\substack{V\subseteq\overline{T\cup J(v_{i-1})},\\|U|+|V|=t}}a_{T\cup U\cup V}\cdot\vbm_V^{(i)}\\
&=\sum_{r=1}^t\sum_{\substack{U\subseteq J(v_{i-1})\cap\overline{T},\\|U|=r}}\vbm_U^{(i)}\sum_{\substack{V\subseteq\overline{T\cup J(v_{i-1})},\\|U|+|V|=t}}a_{T\cup U\cup V}\cdot\vbm_V^{(i-1)}
\tag{since $\vbm_j^{(i)}=\vbm_j^{(i-1)}$ for all $j\notin J(v_{i-1})$}\\
&=\sum_{r=1}^t
\underbrace{\sum_{\substack{U\subseteq J(v_{i-1})\cap\overline{T},\\|U|=r}}\vbm_U^{(i)}\sum_{\substack{V\subseteq\overline{T\cup U},\\|U|+|V|=t}}a_{T\cup U\cup V}\cdot\vbm_V^{(i-1)}}_{A(T,r,i)}.
\tag{since $\vbm_j^{(i-1)}=0$ for all $j\in J(v_{i-1})$}
\end{align*}
Observe that conditioning on $v_{i-1}$,
\begin{itemize}
\item if $r$ is an even number, then $A(T,r,i)$ is a fixed value independent of $\vbm^{(i)}$; 
\item if $r$ is an odd number, then $A(T,r,i)$ is an unbiased coin with magnitude independent of $\vbm^{(i)}$.
\end{itemize}
Therefore, trying to apply \Cref{lem:adaptive_azuma}, we write $X_T^{(i)}-X_T^{(i-1)}=\mu_T^{(i)}+\Delta_T^{(i)}\cdot z_T^{(i)}$, where $z_T^{(1)},\ldots,z_T^{(D)}$ are independent unbiased coins in $\binpm$ and $\mu_T^{(i)}=\Delta_T^{(i)}=0$ for $D'<i\le D$ and
\begin{equation}\label{eq:def_mu_delta}
\mu_T^{(i)}=\sum_{\substack{r=2,\\\text{even}}}^tA(T,r,i)
\quad\text{and}\quad
\Delta_T^{(i)}=\abs{\sum_{\substack{r=1,\\\text{odd}}}^tA(T,r,i)}
\quad\text{for }1\le i\le D'.
\end{equation}

\paragraph*{\fbox{First Bound on $A(T,r,i)$.}}
Let $\Ecal_1$ be the following event:
\begin{align*}
\Ecal_1=\text{`` }\exists \hat t\in\cbra{0,\ldots,t-1},\exists T'\in\binom{[n]}{\ell-\hat t},\exists i'\in[D],~\abs{X_{T'}^{(i')}}\ge M\pbra{D,k,\ell,\hat t,\eps}\text{''}.
\end{align*}
By the induction hypothesis, we have 
\begin{equation}\label{eq:level_l_gamma}
\Pr\sbra{\Ecal_1}\le (t-1)\cdot \eps.
\end{equation}
We first derive a simple bound, that will be effective for small values of $|J(v_{i-1})|$.
\begin{claim}\label{clm:level_l_first}
When $\Ecal_1$ does not happen, $\abs{A(T,r,i)}\le\abs{J(v_{i-1})}^r\cdot M(D,d,k,\ell,t-r,\eps)$ holds for all $r\in[t],i\in[D],T\in\binom{[n]}{\ell-t}$.
\end{claim}
\begin{proof}
Since $\Ecal_1$ does not happen, by union bound we have
\begin{align*}
\abs{A(T,r,i)}
&=\abs{\sum_{\substack{U\subseteq J(v_{i-1})\cap\overline{T},\\|U|=r}}\vbm_U^{(i)}\sum_{\substack{V\subseteq\overline{T\cup U},\\|U|+|V|=t}}a_{T\cup U\cup V}\cdot\vbm_V^{(i-1)}}
\le\abs{J(v_{i-1})}^r\max_{U\subseteq\overline T,|U|=r}\abs{X_{T\cup U}^{(i-1)}}\\
&\le\abs{J(v_{i-1})}^r\cdot M(D,d,k,\ell,t-r,\eps).
\tag*{\qedhere}
\end{align*}
\end{proof}

\paragraph*{\fbox{Second Bound on $A(T,r,i)$.}}
The second bound requires a more refined decomposition on $A(T,r,i)$.

Assume that $c(i-1)$ is the index of $C(v_{i-1})$ in $v_0,\ldots,v_{D'}$, i.e., $v_{c(i-1)}=C(v_{i-1})$. This means that $v_{c{(i-1)}}$ is the closest ancestor to $v_{i-1}$ that is $2k$-clean.
Then define 
$$
L(v_{i-1})=\bigcup_{c(i-1)\le i'<i-1}J(v_{i'}).
$$
The elements of $L(v_{i-1})$ are precisely the coordinates fixed by the queries from $Q_{v_{c(i-1)}}$ to $Q_{v_{i-1}}$, excluding the latter. Since $\Tcal_{C(v_{i-1})}$ makes non-adaptive queries before (and possibly even after) reaching $v_i$, $L(v_{i-1})$ and $J(v_{i-1})$ depend only on $C(v_{i-1})$ and $i$.
We now expand $A(T,r,i)$ by also grouping terms based on the number of coordinates in $L(v_{i-1})$ as follows:
\begin{align*}
A(T,r,i)
&=\sum_{\substack{U\subseteq J(v_{i-1})\cap\overline{T},\\|U|=r}}\vbm_U^{(i)}\sum_{\substack{V\subseteq\overline{T\cup U},\\|U|+|V|=t}}a_{T\cup U\cup V}\cdot\vbm_V^{(i-1)}\\
&=\sum_{r'=0}^{t-r}\sum_{\substack{U\subseteq J(v_{i-1})\cap\overline{T},\\|U|=r}}\vbm_U^{(i)}
\sum_{\substack{W\subseteq L(v_{i-1})\cap\overline T,\\|W|=r'}}\vbm_W^{(i-1)}\sum_{\substack{W'\subseteq\overline{T\cup U\cup L(v_{i-1})}\\|W'|=t-r-r'}}a_{T\cup U\cup W\cup W'}\cdot\vbm_{W'}^{(i-1)}\\
&=\sum_{r'=0}^{t-r}\sum_{\substack{U\subseteq J(v_{i-1})\cap\overline{T},\\|U|=r}}\vbm_U^{(i)}
\sum_{\substack{W\subseteq L(v_{i-1})\cap\overline T,\\|W|=r'}}\vbm_W^{(i-1)}\sum_{\substack{W'\subseteq\overline{T\cup U\cup L(v_{i-1})}\\|W'|=t-r-r'}}a_{T\cup U\cup W\cup W'}\cdot\vbm_{W'}^{c(i-1)}
\tag{since $\vbm_j^{(i-1)}=\vbm_j^{c(i-1)}$ for all $j\notin L(v_{i-1})$}\\
&=\sum_{r'=0}^{t-r}\sum_{\substack{U\subseteq J(v_{i-1})\cap\overline{T},\\|U|=r}}\vbm_U^{(i)}
\sum_{\substack{W\subseteq L(v_{i-1})\cap\overline T,\\|W|=r'}}\vbm_W^{(i-1)}\sum_{\substack{W'\subseteq\overline{T\cup U\cup W}\\|W'|=t-r-r'}}a_{T\cup U\cup W\cup W'}\cdot\vbm_{W'}^{c(i-1)}
\tag{since $\vbm_j^{c(i-1)}=0$ for all $j\in L(v_{i-1})$}\\
&=\sum_{r'=0}^{t-r}\underbrace{
\sum_{\substack{U\subseteq J(v_{i-1})\cap\overline T,\\|U|=r}}\vbm_U^{(i)}
\sum_{\substack{W\subseteq L(v_{i-1})\cap\overline T,\\|W|=r'}}\vbm_W^{(i-1)}
\cdot X_{T\cup U\cup W}^{c(i-1)}}_{\Gamma_T^{(i)}(r,r')}.
\end{align*}

Since $C(v_{i-1})$ is $2k$-clean, by \Cref{fct:coefficients_of_Pv}, the collection of random variables
$$
\cbra{\vbm_j^{(i)}\mid j\in J(v_{i-1})}\cup\cbra{\vbm_j^{(i-1)}\mid j\in L(v_{i-1})}
$$ 
is $2k$-wise independent conditioning on $C(v_{i-1})$. Note that $\Gamma_T^{(i)}(r,r')$ is a polynomial of degree at most $r+r'\le\ell<k$, that $\E\sbra{\Gamma_T^{(i)}(r,r')\mid C(v_{i-1})}=0$, and
\begin{align*}
\sigma^2_T(r,r',C(v_{i-1}),i)
&:=\E\sbra{\pbra{\Gamma_T^{(i)}(r,r')}^2\mid C(v_{i-1})}
=\sum_{\substack{U\subseteq J(v_{i-1})\cap\overline T,\\|U|=r}}
\sum_{\substack{W\subseteq L(v_{i-1})\cap\overline T,\\|W|=r'}}
\pbra{X_{T\cup U\cup W}^{c(i-1)}}^2\\
&\le\pbra{|J(v_{i-1})|}^r\pbra{|L(v_{i-1})|}^{r'}\pbra{\max_{|T'|=r+r'+\ell-t,i'\in[D]}\abs{X_{T'}^{(i')}}}^2\\
&\le\pbra{|J(v_{i-1})|}^rD^{r'}\pbra{\max_{|T'|=r+r'+\ell-t,i'\in[D]}\abs{X_{T'}^{(i')}}}^2.
\tag{since $|L(v_{i-1})|\le D$ by \Cref{fct:depth_to_singleton}}
\end{align*}
We also have the following claim, the proof of which follows from \Cref{lem:low_degree_polys} applied to the low degree polynomial $\Gamma_T^{(i)}$. The proof is deferred to \Cref{app:clm:level_l_Gamma}. 
\begin{claim}\label{clm:level_l_Gamma}
$\Pr\sbra{\Ecal_2}\le\eps/3$, where $\Ecal_2$ is the following event:
$$
\text{`` }\exists T\in\binom{[n]}{\ell-t},i,r,r',\abs{\Gamma_T^{(i)}(r,r')}\ge\pbra{100\min\cbra{k,\log\pbra{\tfrac{n^\ell}\eps}}\cdot\pbra{\tfrac{n^\ell}\eps}^{\frac6k}}^{\frac{r+r'}2}\cdot\sigma_T(r,r',C(v_{i-1}),i)\text{''}.
$$
\end{claim}

On the other hand, when $\Ecal_1\lor\Ecal_2$ does not happen, the following calculation holds for all $T\in\binom{[n]}{\ell-t}$, $i\in[D']$, $r\in[t]$, $0\le r'\le t-r$:
\begin{align*}
\abs{\Gamma_T^{(i)}(r,r')}
&\le M\pbra{D,k,\ell,t-r-r',\eps }\cdot\sqrt{\pbra{100\min\cbra{k,\log\pbra{\tfrac{n^\ell}\eps}}\cdot\pbra{\tfrac{n^\ell}\eps}^{\frac6k}}^{r+r'}\pbra{|J(v_{i-1})|}^r\cdot D^{r'}}\\
&\le M\pbra{D,k,\ell,t-r-r',\eps }\cdot\sqrt{\pbra{100\cdot\pbra{\tfrac{n^\ell}\eps}^{\frac6k}}^{r+r'}\pbra{|J(v_{i-1})|\cdot k}^r\cdot \pbra{D\cdot\log\pbra{\tfrac{n^\ell}\eps}}^{r'}}\\
&=\sqrt{\pbra{\tau (D+dk)\pbra{\tfrac{n^\ell}\eps}^{\frac6k}\log\pbra{\tfrac{n^\ell}\eps}}^{t-r-r'}\hspace{-8pt}\pbra{100\pbra{\tfrac{n^\ell}\eps}^{\frac6k}}^{r+r'}\hspace{-8pt}\pbra{|J(v_{i-1})|\cdot k}^r\pbra{D\cdot\log\pbra{\tfrac{n^\ell}\eps}}^{r'}}\\
&\le\sqrt{\pbra{\tau (D+dk)\pbra{\tfrac{n^\ell}\eps}^{\frac6k}\log\pbra{\tfrac{n^\ell}\eps}}^t
\pbra{\tfrac{100}\tau}^{r+r'}\pbra{\tfrac{|J(v_{i-1})|}{d\cdot\log\pbra{n^\ell/\eps}}}^r}\\
&\le\sqrt{\pbra{\tau (D+dk)\pbra{\tfrac{n^\ell}\eps}^{\frac6k}\log\pbra{\tfrac{n^\ell}\eps}}^t
\pbra{\tfrac{200}\tau}^{r+r'}\pbra{\tfrac{|J(v_{i-1})|}{2d}}^r\tfrac1{\log\pbra{n^\ell/\eps}}}\\
&=M(D,d,k,\ell,t,\eps)\cdot\sqrt{\pbra{\tfrac{200}\tau}^{r+r'}\pbra{\tfrac{|J(v_{i-1})|}{2d}}^r\tfrac1{\log\pbra{n^\ell/\eps}}}.
\end{align*}
Hence we have a second bound on $A(T,r,i)$.
\begin{claim}\label{clm:level_l_second}
When $\Ecal_1\lor\Ecal_2$ does not happen, the following holds for all $r\in[t],i\in[D],T\in\binom{[n]}{\ell-t}$:
$$
\abs{A(T,r,i)}\le\frac{M(D,d,k,\ell,t,\eps)}{\sqrt{\log\pbra{n^\ell/\eps}}}\cdot\sqrt{\pbra{\frac{800}\tau}^r\pbra{\frac{|J(v_{i-1})|}{2d}}^r}.
$$
\end{claim}
\begin{proof}
Since $\Ecal_1\lor\Ecal_2$ does not happen, by union bound and noticing $\tau\ge 800$ we have
\begin{align*}
\abs{A(T,r,i)}
&\le\sum_{r'=0}^{t-r}\abs{\Gamma_T^{(i)}(r,r')}
\le \frac{M(D,d,k,\ell,t,\eps)}{\sqrt{\log\pbra{n^\ell/\eps}}} \cdot\sqrt{\pbra{\frac{200}\tau}^r\pbra{\frac{|J(v_{i-1})|}{2d}}^r} \cdot\sum_{r'=0}^{+\infty}\pbra{\frac{200}\tau}^{r'/2}\\
&\le \frac{M(D,d,k,\ell,t,\eps)}{\sqrt{\log\pbra{n^\ell/\eps}}} \cdot\sqrt{\pbra{\frac{800}\tau}^r\pbra{\frac{|J(v_{i-1})|}{2d}}^r}.
\tag*{\qedhere}
\end{align*}
\end{proof}

\paragraph*{\fbox{Final Bound on $\mu_T^{(i)}$ and $\delta_T^{(i)}$.}}
Combining \Cref{clm:level_l_first} and \Cref{clm:level_l_second}, if $\Ecal_1\lor\Ecal_2$ does not happen we have
\begin{equation}\label{eq:final_bound_on_A}
\abs{A(T,r,i)}\le M(D,d,k,\ell,t-r,\eps)+\frac{M(D,d,k,\ell,t,\eps)}{\sqrt{\log\pbra{n^\ell/\eps}}}\cdot\sqrt{\pbra{\frac{800}\tau}^r\pbra{\frac{\abs{J(v_{i-1})}}{2d}}^r} \cdot \indicator_{|J(v_{i-1})|>1}
\end{equation}
To see this, if $\abs{J(v_{i-1})}\le 1$, we use the bound from \Cref{clm:level_l_first} as the first term in \Cref{eq:final_bound_on_A}. Otherwise $\abs{J(v_{i-1})}> 1$, in which case we use the bound from \Cref{clm:level_l_second} as the second term in \Cref{eq:final_bound_on_A}.

By \Cref{cor:depth_to_singleton}, we can now bound $\sum_{i=1}^D\abs{\mu_T^{(i)}}$ and $\sum_{i=1}^D\abs{\Delta_T^{(i)}}^2$ as \Cref{clm:level_l_final}. Its proof is deferred in \Cref{app:clm:level_l_final}.
\begin{claim}\label{clm:level_l_final}
When $\Ecal_1\lor\Ecal_2$ does not happen, $\sum_{i=1}^D\abs{\mu_T^{(i)}}\le R$ and $\sum_{i=1}^D\abs{\Delta_T^{(i)}}^2\le R^2$ hold for all $T\in\binom{[n]}{\ell-t}$, where 
\begin{equation}\label{eq:def_R}
R=\frac{M(D,d,k,\ell,t,\eps)}{5\cdot\sqrt{\log\pbra{n^\ell/\eps}}}.
\end{equation}
\end{claim}

\paragraph*{\fbox{Complete Induction.}}
Let $\beta=\sqrt{2\cdot\log\pbra{n^\ell/\eps}}\ge1$ and observe that
\begin{align*}
R+\beta\cdot\sqrt2\cdot R
&\le\beta\cdot2\sqrt2\cdot R
\tag{due to $\beta\ge1$}\\
&=\frac{2\sqrt2\cdot\sqrt{2\cdot\log\pbra{n^\ell/\eps}}}{5\cdot\sqrt{\log\pbra{n^\ell/\eps}}}
\cdot M(D,d,k,\ell,t,\eps)
\tag{due to \Cref{eq:def_R}}\\
&\le M(D,d,k,\ell,t,\eps).
\end{align*}

Then we have 
\begin{align*}
&\Pr\sbra{\exists t'\in\cbra{0,\ldots,t},\exists T'\in\binom{[n]}{\ell-t'},\exists i\in[D],~\abs{X_T^{(i)}}\ge M\pbra{D,d,k,\ell,t',\eps}}\\
&=\Pr\sbra{\Ecal_1\bigvee\pbra{\exists T\in\binom{[n]}{\ell-t},\exists i\in[D],~\abs{X_T^{(i)}}\ge M\pbra{D,d,k,\ell,t,\eps}}}\\
&\le\Pr\sbra{\pbra{\Ecal_1\lor\Ecal_2}\bigvee\pbra{\exists T\in\binom{[n]}{\ell-t},\exists i\in[D],~\abs{X_T^{(i)}}\ge R+\beta\cdot\sqrt2\cdot R}}\\
&\le (t-1)\cdot \eps + \frac{\eps}3+2n^{\ell-t}\cdot e^{-\beta^2/2}
\tag{due to \Cref{eq:level_l_gamma}, \Cref{clm:level_l_Gamma}, \Cref{lem:adaptive_azuma}, and \Cref{clm:level_l_final}}\\
&\le (t-1)\cdot \eps  +\frac{\eps}3+\frac13\cdot n^\ell\cdot e^{-\beta^2/2}\\
&\le t\cdot\eps.
\tag*{\qedhere}
\end{align*}
\end{proof}

Before we prove the complete level-$\ell$ bound for parity decision trees, we first prove a simple bound for the number of vectors with a given weight in a subspace.
\begin{lemma}\label{lem:weight_l_vectors_in_dim_d}
Let $\ell\ge1$ be an integer and $\Scal$ be a subspace of rank at most $d$. 
Let $U=\cbra{S\mid |S|=\ell,S\in\Scal}$, then $\abs{U}\le\min\cbra{\binom{d\cdot\ell}\ell,2^d-1}$.
\end{lemma}
\begin{proof}
Let $\cbra{S_1,\ldots,S_{d'}}$ be a maximal set of independent vectors in $U$. Then $d'\le d$ and $|S_i|=\ell$ holds for all $i\in[d']$.
Since $U\subseteq\Span\abra{S_1,\ldots,S_{d'}}$ and $\emptyset\notin U$, we have 
$$
|U|\le\abs{\Span\abra{S_1,\ldots,S_{d'}}}-1=2^{d'}-1\le2^d-1.
$$
On the other hand, observe that $U\subseteq\binom{S_1\cup\cdots\cup S_{d'}}\ell$, hence we also have
\begin{equation*}
|U|\le\abs{\binom{S_1\cup\cdots\cup S_{d'}}\ell}\le\binom{d'\cdot\ell}\ell\le\binom{d\cdot\ell}\ell.
\tag*{\qedhere}
\end{equation*}
\end{proof}

We remark that in \Cref{lem:weight_l_vectors_in_dim_d}, it is conjectured the bound should be $\binom{d+1}\ell$ when $d\ge2\cdot\ell$ \cite{DBLP:journals/combinatorics/Kramer10,briggs2018extremal}.

\begin{theorem}\label{thm:level_l}
Let $\ell\ge1$ be an integer.
Let $\Tcal\colon\binpm^n\to\bin$ be a depth-$d$ parity decision tree where $n\ge\max\cbra{d,\ell}$. Let $p=\Pr\sbra{\Tcal(x)=1}\ge2^{-d}$.\footnote{If $p<2^{-d}$, then $p=0$ and $\Tcal\equiv0$.} Then we have
$$
\sum_{S\subseteq[n]:|S|=\ell}\abs{\widehat\Tcal(S)}
\le p\cdot\min\cbra{\binom{d\cdot\ell}\ell,2^d-1,\;O\!\pbra{\sqrt d\cdot\log\pbra{\tfrac{n^\ell}p}}^\ell}=O\!\pbra{\sqrt d\cdot\ell\cdot\log(n)}^\ell.
$$
\end{theorem}
\begin{proof}
For any $S\in\binom{[n]}\ell$, let $a_S=\sgn\pbra{\widehat\Tcal(S)}$. Now we prove the bounds separately.

\paragraph*{\fbox{First Two Bounds.}}
Let $v_0,\ldots,v_{d'}$ be a random root-to-leaf path.
Then by the definition of $\widehat{\Pcal_v}$ and $\Scal_v$ and \Cref{fct:coefficients_of_Pv}, we have
\begin{align}
\sum_S\abs{\widehat\Tcal(S)}
&=\sum_Sa_S\cdot\widehat\Tcal(S)
=\E_{v_0,\ldots,v_{d'}}\sbra{\Tcal(v_{d'})\cdot\sum_Sa_S\cdot\widehat{\Pcal_{v_{d'}}}(S)}\notag\\
&\le\E_{v_0,\ldots,v_{d'}}\sbra{\Tcal(v_{d'})\cdot\sum_S\abs{\widehat{\Pcal_{v_{d'}}}(S)}}
=\E_{v_0,\ldots,v_{d'}}\sbra{\Tcal(v_{d'})\cdot|V|},\label{eq:level_l_thm_1}
\end{align}
where $a_S = \sgn\pbra{\widehat\Tcal(S)}$ and $V=\cbra{S\in\binom{[n]}\ell\mid S\in\Scal_{v_{d'}}}$.
Note that 
$$
\rank\pbra{\Scal_{v_{d'}}}=\rank\pbra{\Span\abra{Q_{v_0},\ldots,Q_{v_{d'-1}}}}\le d'\le d.
$$
Hence by \Cref{lem:weight_l_vectors_in_dim_d}, we have $\Cref{eq:level_l_thm_1}\le\min\cbra{\binom{d\cdot\ell}\ell,2^d-1}\cdot\E\sbra{\Tcal(v_{d'})}=p\cdot\min\cbra{\binom{d\cdot\ell}\ell,2^d-1}$.

\paragraph*{\fbox{Third Bound.}}
By \Cref{lem:clean_PDT}, we construct a $2k$-clean parity decision tree $\Tcal'$ of depth $D\le2d\cdot k$ equivalent to $\Tcal$, where $k=\Theta\pbra{\log\pbra{n^\ell/p}}\ge4\cdot\ell$. We also add dummy variables to make sure $n'=\max\cbra{\tau,k,6D,n}$, where $\Tcal'$ has $n'$ inputs and $\tau$ is the universal constant in \Cref{lem:level_l}.

Let $u_0,\ldots,u_{D'}$ be a random root-to-leaf path in $\Tcal'$.
Define $\ubm^{(0)},\ldots,\ubm^{(D')}\in\cbra{-1,0,+1}^n$ by setting $\ubm^{(i)}_j=\widehat{\Pcal_{u_i}}(j)$ for each $0\le i\le D'$ and $j\in[n]$.
Then extend $\ubm^{(D'+1)}=\ubm^{(D'+2)}=\cdots=\ubm^{(D)}$ to equal $\ubm^{(D')}$.
By \Cref{lem:clean_to_Fourier}, we have
\begin{equation}\label{eq:level_l_thm_2}
\sum_S\abs{\widehat\Tcal(S)}
=\sum_S\abs{\widehat{\Tcal'}(S)}
=\E_{u_0,\ldots,u_{D'}}\sbra{\Tcal(u_{D'})\cdot\sum_Sa_S\cdot\ubm_S^{(D)}}
\le\E_{u_0,\ldots,u_{D'}}\sbra{\Tcal(u_{D'})\cdot|U|},
\end{equation}
where $U=\sum_Sa_S\cdot\ubm_S^{(D)}$.

Now we apply \Cref{lem:level_l} with $t=\ell,\eps=\Theta\pbra{p/d^{\ell/2}}\le1/2$ to obtain the following bound\footnote{Since $n\ge\max\cbra{\ell,d}$, we know $k=\Theta\pbra{\log\pbra{n^\ell/p}}=O(n^2)$ and $D\le 2d\cdot k=O(n^3)$. Hence $n'=\max\cbra{\tau,k,6D,n}=O(n^3)$. Also $n^\ell/\eps\le n^{O(\ell)}/p$ and by our choice of $k=\Theta\pbra{\log(n^{\ell}/p)}$ we have  $\pbra{n^\ell/\eps}^{6/k} = O(1)$.}
$$
M=M(D,d,k,\ell,\ell,\eps)=\pbra{O\pbra{\sqrt d\cdot\log\pbra{\tfrac{n^\ell}p}}}^\ell
$$ 
such that $\Pr\sbra{|U|\ge M}\le \ell \cdot \eps$.
Then, combining the first bound, we have
\begin{align*}
\Cref{eq:level_l_thm_2}
&=\E\sbra{\Tcal(u_{D'})\cdot|U|\cdot\pbra{\indicator_{|U|<M}+\indicator_{|U|\ge M}}}
\le M\cdot\E\sbra{\Tcal(u_{D'})}+
\ell \cdot \eps\cdot\binom{d\cdot\ell}\ell\\
&=p\cdot\pbra{O\pbra{\sqrt d\cdot\log\pbra{\tfrac{n^\ell}p}}}^\ell,
\end{align*}
which is maximized at $p=1$, hence $\Cref{eq:level_l_thm_2}=O\!\pbra{\sqrt d\cdot\ell\cdot\log(n)}^\ell$ as desired.
\end{proof}

\section{Fourier Bounds for Noisy Decision Trees}\label{sec:noisy}

Let $\Tcal$ be a noisy decision tree. By adding queries with zero correlation, we assume without loss of generality each root-to-leaf path in the noisy decision tree is of the same length. Let $v$ be any node of $\Tcal$. We use $\Pcal_v$ to denote the uniform distribution over $\binpm^n$ \emph{conditioning} on reaching $v$. Note that $\Pcal_v$ is always a \emph{product distribution}. As before, for any $S\subseteq[n]$ we define $\widehat{\Pcal_v}(S)=\E_{x\sim\Pcal_v}\sbra{x_S}$.

\begin{claim}\label{clm:noisedelta}
Let $\Tcal\colon\binpm^n\to \bin$ be a cost-$d$ noisy decision tree. Let $v_0,\ldots,v_D$ be any root-to-leaf path in $\Tcal$. Define $\vbm^{(0)},\ldots,\vbm^{(D)}\in[-1,1]^n$ by setting $\vbm^{(i)}_j=\widehat{\Pcal_{v_i}}(j)$ for each $0\le i\le D$ and $j\in[n]$. Then for any $i\in \{0,\ldots,D-1\}$, $\vbm_{q_{v_i}}^{(i+1)}-\vbm_{q_{v_i}}^{(i)}$ is a mean-zero random variable with magnitude bounded by $2\cdot \abs{\gamma_{v_i}}$.
\end{claim}
\begin{proof}
Fix $i\in\{0,\ldots,D-1\}$. For convenience, let $j=q_{v_i}$, $\gamma=\gamma_{v_i}$, and $\alpha=\vbm_j^{(i)}$. 
Suppose $\abs{\gamma}=1$ then $\abs{v_j^{(i+1)}-v_j^{(i)}}\le 2 = 2\cdot\abs{ \gamma_{v_i}}$ as desired. Now we turn to the case $\abs{\gamma}<1$. 

Note that for the distribution $\Pcal_{v_i}$, the measure of $x_j=1$ (resp., $x_j=-1$) inputs is $(1+\alpha)/2$ (resp., $(1-\alpha)/2$).
The measure of $x_j=1$ (resp., $x_j=-1$) inputs that follow the edge labeled $1$ is $a:=(1+\alpha)(1+\gamma)/4$ (resp., $b:=(1-\alpha)(1-\gamma)/4$). The total measure of inputs that take the edge labeled $1$ is $a+b$ and the resulting node $v_{i+1}$ satisfies $\vbm^{(i+1)}_j=(a-b)/(a+b)$. This implies that
\[ 
\vbm^{(i+1)}_j = \begin{cases}  
\frac{\alpha + \gamma}{1 + \gamma\cdot\alpha} & \text{ with probability } \frac{1+\gamma\cdot\alpha}2, \\  
\frac{\alpha- \gamma}{1 - \gamma\cdot\alpha} & \text{ with probability } \frac{1-\gamma\cdot\alpha}2. 
\end{cases}
\]
The above calculation implies
\[ 
\vbm^{(i+1)}_j-\vbm^{(i)}_j = \begin{cases}  
\gamma\cdot  \frac{1-\alpha^2}{1 + \gamma\cdot \alpha} & \text{ with probability } \frac{1+\gamma\cdot \alpha}2, \\ 
-\gamma \cdot \frac{1-\alpha^2 }{1 - \gamma\cdot \alpha} & \text{ with probability } \frac{1-\gamma\cdot \alpha}2,
\end{cases} 
\]
and thus $\vbm^{(i+1)}_j-\vbm^{(i)}_j$ is a mean-zero random variable.
Since $\alpha\in[-1,1]$ and $\gamma\in(-1,1)$, we have 
\[
\max\cbra{ \frac{1-\alpha^2}{1-\gamma\cdot\alpha}, \frac{1-\alpha^2}{1+\gamma\cdot\alpha} }\le \frac{1-\alpha^2}{1-\abs{\alpha}}=1+\abs{\alpha}\le 2,
\]
which implies $\abs{ \vbm^{(i+1)}_j-\vbm^{(i)}_j } \le 2\cdot \abs{\gamma} $.  
\end{proof}

We now prove the general Fourier bounds. As before, for any $S\subseteq[n]$, let $\vbm_S^{(i)}$ be $\prod_{j\in S}\vbm^{(i)}_j$.

\begin{lemma}\label{lem:noisylevell} 
There exists a universal constant $\tau$ such that the following holds. Let $\ell\ge1$ be an integer. Let $\Tcal\colon\binpm^n\to\bin$ be a cost-$d$ noisy decision tree. 

Let $v_0,\ldots,v_D$ be a random root-to-leaf path in $\Tcal$. Define $\vbm^{(0)},\ldots,\vbm^{(D)}\in[-1,1]^n$ by setting $\vbm^{(i)}_j=\widehat{\Pcal_{v_i}}(j)$ for each $0\le i\le D$ and $j\in[n]$. Then for any sequence $a_S\in\cbra{-1,0,1},S\in\binom{[n]}\ell$, any $\eps\le1/2$ and $t\in\cbra{0,\ldots,\ell}$, we have
\[ 
\Pr\sbra{ \exists T\in\binom{[n]}{\ell-t}, \exists i\in [D],\abs{\sum_{S\subseteq \overline{T},|S|=t} a_{S\cup T}\cdot \vbm_S^{(i)}} \ge S(d,\ell,t,\eps) } \le \eps\cdot t,
\]
where $S(d,\ell,0,\eps)=1$ and
\[ 
S(d,\ell,t,\eps)=\sqrt{\pbra{\tau\cdot d}^t\cdot \log\pbra{\tfrac{n^{\ell-t}}\eps} \cdots \log\pbra{\tfrac{n^{\ell-1}}\eps}}
\qquad\text{for $t\in[\ell]$}. 
\]
\end{lemma}
\begin{proof}
We prove the bound by induction on $t$ and show $\tau=32$ suffices. The base case $t=0$ is trivial, since for any $T$ of size $\ell$ and any $i$, we have $\abs{ a_T \cdot v_{\emptyset}^{(i)}}\le 1 = S(d,\ell,0,\eps)$. 
 
Now we focus on the case $1\le t\le \ell$. For any $T\in\binom{[n]}{\le\ell}$, define $X_T^{(0)},\ldots,X_T^{(D)}$ by $X^{(i)}_T=\sum_{S\subseteq \overline T, |S|+|T|=\ell} a_{S\cup T}\cdot \vbm_S^{(i)}$. Define $\delta^{(i)}_T$ for $i\in[D]$ as follows:
\begin{align*}
\delta_T^{(i)}=X_T^{(i)}-X_T^{(i-1)} 
&= \sum_{S\subseteq \overline{T}, |S|=t, S\ni q_{v_{i-1}}}   a_{S\cup T}\cdot  \pbra{\vbm^{(i)}_S-\vbm^{(i-1)}_S } \\
&=  \pbra{\vbm^{(i)}_{q_{v_{i-1}}}-\vbm^{(i-1)}_{q_{v_{i-1}}} } \cdot \sum_{S'\subseteq \overline{T\cup\{q_{v_{i-1}}\}}, |S'|=t-1} a_{S'\cup\cbra{q_{v_{i-1}}}\cup T}\cdot\vbm_S^{(i-1)}\\
&=  \pbra{\vbm^{(i)}_{q_{v_{i-1}}}-\vbm^{(i-1)}_{q_{v_{i-1}}} } \cdot X^{(i-1)}_{T\cup \{q_{v_{i-1}}\} }.
\end{align*}
Note that by \Cref{clm:noisedelta} and conditioning on $v_{i-1}$, $\delta_T^{(i)}$ is a mean-zero random variable.

The induction hypothesis implies that with all but $\eps\cdot (t-1)$ probability, for all $i\in [D]$ and $T'\in\binom{[n]}{\ell-t+1}$, we have $\abs{X_{T' }^{(i)}} \le S(d,\ell,t-1,\eps)$. By \Cref{clm:noisedelta}, we have
$$
\abs{\delta_T^{(i)}}
=\abs{\vbm^{(i)}_{q_{v_{i-1}}}-\vbm^{(i-1)}_{q_{v_{i-1}}}}\cdot \abs{X_{T\cup\cbra{q_{v_{i-1}}}}^{(i-1)}}
\le2\cdot\abs{\gamma_{v_{i-1}}}\cdot S(d,\ell,t-1,\eps).
$$
Denote by $\Delta_{T}^{(i)} = 2\cdot\abs{\gamma_{v_{i-1}}}\cdot S(d,\ell,t-1,\eps)$. We can thus express $X_T^{(i)} = X_T^{(i-1)}  + \Delta_T^{(i)} \cdot z_T^{(i)}$ where $\abs{z_T^{(i)}}\le 1$.
Then we apply \Cref{lem:adaptive_azuma} to the family of martingales $X_T^{(0)},\ldots,X_T^{(D)},|T|\in\binom{[n]}{\ell-t}$ with difference sequence $\delta_T^{(i)}=\Delta_T^{(i)}\cdot z_T^{(i)}$ satisfying 
$$
\sum_{i=1}^D\pbra{\Delta_T^{(i)}}^2=4\cdot( S(d,\ell,t-1,\eps))^2\cdot\sum_{i=1}^D\abs{\gamma_{v_{i-1}}}^2\le4d\cdot\pbra{S(d,\ell,t-1,\eps)}^2.
$$ 
Hence for any $\beta\ge0$, we have
\[ 
\Pr\sbra{ \exists T\in\binom{[n]}{\ell-t},\exists i\in [D], \abs{X_T^{(i)}} \ge 2\beta\cdot \sqrt{2d}\cdot S(d,\ell, t-1,\eps)} \le \eps\cdot(t-1)+ 2\cdot n^{\ell-t}\cdot e^{-\beta^2/2}. 
\]
Since $\eps\le1/2$, we can set $\beta=2\cdot \sqrt{\log(n^{\ell-t}/\eps)}$ so that $2\cdot n^{\ell-t}\cdot e^{-\beta^2/2}\le\eps$, which completes the induction by noticing 
\begin{equation*}
2\beta\cdot\sqrt{2d}\cdot S(d,\ell,t-1,\eps)
=\sqrt{32\cdot d\cdot\log\pbra{\tfrac{n^{\ell-t}}\eps}}\cdot S(d,\ell,t-1,\eps)
\le S(d,\ell,t,\eps).
\tag*{\qedhere}
\end{equation*}
\end{proof}

\begin{theorem}\label{thm:level-l_noisy_DT}
Let $\ell\ge1$ and $n\ge\max\cbra{\ell,2}$ be integers. Let $\Tcal\colon\binpm^n\to\bin$ be a cost-$d$ noisy decision tree. Let $p=\Pr[\Tcal(x)=1]\in(0,1/2]$.\footnote{If $p>1/2$, then we can consider $\tilde\Tcal=1-\Tcal$ by symmetry.} Then we have 
\[ 
\sum_{S\subseteq[n],|S|=\ell} \abs{\widehat{\Tcal}(S)} 
\le p\cdot O(d)^{\ell/2}\cdot\sqrt{\log\pbra{\tfrac1p}\pbra{\log\pbra{\tfrac{n^\ell}p}}^{\ell-1}}
=O(d)^{\ell/2}\cdot\sqrt{1+\pbra{\ell\log(n)}^{\ell-1}}. 
\]
\end{theorem}
\begin{proof}
For any $S\in\binom{[n]}\ell$, let $a_S=\sgn\pbra{\widehat{\Tcal}(S)}$.
Let $v_0,\ldots,v_D$ be a random root-to-leaf path in $\Tcal$. 
Note that
\begin{equation}\label{eq:noisyeqn}  \sum_S\abs{\widehat{\Tcal}(S)}= \sum_S a_S\cdot \widehat{\Tcal}(S) = \E \sbra{ \Tcal(v_D) \cdot \sum_S a_S \cdot \vbm_S^{(D)} }  \le  \E \sbra{ \Tcal(v_D) \cdot \abs{V}},
\end{equation}
where $V=\sum_S a_S\cdot_S\vbm_S^{(D)}$. 
By \Cref{lem:noisylevell}, we know $ \Pr\sbra{  \abs{V} \ge S(\eps) } \le \eps \cdot\ell$, where
\[ 
S(\eps)=S(d,\ell,\ell,\eps) = \sqrt{O(d)^\ell \cdot \log\pbra{\tfrac{n^{\ell-1}}{\eps}}\cdots \log\pbra{\tfrac{n^0}{\eps}} } \le \sqrt{O(d)^\ell \cdot \pbra{\log\pbra{\tfrac{n^{\ell-1}}{\eps}}}^{\ell-1}\cdot \log\pbra{\tfrac{1}{\eps}}}.
\]

For integer $i\ge1$, let $I_{i}=\sbra{S\pbra{ p/\pbra{\ell2^i}},S\pbra{p/\pbra{\ell2^{i+1}}}}$ and $I_0=\sbra{0,S(p/\ell )}$ be intervals. Then for each $i\ge1$, $\Pr\sbra{\abs{V}\in I_i}\le p/2^i$. We also know that $\E_{v_0,\ldots,v_D}\sbra{\Tcal(v_D)}\le p$. Thus,
\begin{align*} 
\Cref{eq:noisyeqn}
&\le  \E_{v_0,\ldots,v_D}\sbra{\Tcal(v_D)\cdot|V|\cdot \sum_{i=0}^{+\infty} \indicator_{\abs{V}\in I_i}}\\
&\le  S\pbra{\frac p\ell}\cdot\E\sbra{\Tcal(v_D)}+\sum_{i=1}^{+\infty} S\pbra{\frac{p}{\ell\cdot 2^{i+1}}}\cdot  \E\sbra{\indicator_{\abs{V}\in I_i}}\\
&\le\sum_{i=0}^{+\infty} S\pbra{\frac{p}{\ell\cdot 2^{i+1}}}\cdot \frac{p}{2^i}\\
&= \sum_{i=0}^{+\infty} p\cdot \sqrt{ O(d)^\ell \cdot \pbra{\log\pbra{\tfrac {n^{\ell-1}\cdot\ell}p}+i+1}^{\ell-1}\cdot \pbra{\log\pbra{\tfrac1p}+\log(\ell)+i+1}}\cdot \frac{1}{2^i} \\
&\le\sum_{i=0}^{+\infty} p\cdot \sqrt{ O(d)^\ell \cdot \pbra{\pbra{\log\pbra{\tfrac {n^\ell}p}}^{\ell-1}+\pbra{i+1}^{\ell-1}}\cdot \pbra{\log\pbra{\tfrac1p}+i+1}}\cdot \frac{1}{2^i}
\tag{since $n\ge\ell$, and $(x+y)^b\le2^b\cdot\pbra{x^b+y^b}$ and $\sqrt{x+y}\le\sqrt x+\sqrt y$ for $x,y,b\ge0$}\\
&\le p\cdot\sqrt{O(d)^\ell\cdot \log\pbra{\tfrac1p}\pbra{\log\pbra{\tfrac {n^{\ell} }p}}^{\ell-1}},
\end{align*}
where the last inequality follows from $p\le1/2$, $n\ge2$ and 
$$
\sum_{i=0}^{+\infty}(i+1)^{\ell/2}\cdot2^{-i}= O(\ell)^{\ell/2}
\le O(1)^\ell\cdot\ell^{(\ell-1)/2}
\le O(1)^\ell\cdot \pbra{\log\pbra{n^\ell/p}}^{(\ell-1)/2}.
$$
Note that $p\cdot\pbra{\log(1/p)}^k\le O(k)^k$ for $p\in(0,1)$ and $k\ge0$, thus
\begin{align*}
p\cdot\sqrt{\log\pbra{\tfrac1p}\pbra{\log\pbra{\tfrac {n^{\ell} }p}}^{\ell-1}}
&=p\cdot\sqrt{\log\pbra{\tfrac1p}\pbra{\ell\log(n)+\log\pbra{\tfrac1p}}^{\ell-1}}\\
&\le O(1)^\ell\cdot\pbra{\sqrt{\pbra{\ell\log(n)}^{\ell-1}}+\ell^{\ell/2}}\\
&=O(1)^\ell\cdot\sqrt{1+(\ell\log(n))^{\ell-1}}.
\tag*{\qedhere}
\end{align*}
\end{proof}

\section*{Acknowledgement}
We thank anonymous reviewers for helpful comments.

\bibliographystyle{alphaurl} 
\bibliography{ref}

\newcommand{\etalchar}[1]{$^{#1}$}
\begin{thebibliography}{OWZ{\etalchar{+}}14}

\bibitem[AA18]{DBLP:journals/siamcomp/AaronsonA18}
Scott Aaronson and Andris Ambainis.
\newblock Forrelation: {A} problem that optimally separates quantum from
  classical computing.
\newblock {\em {SIAM} J. Comput.}, 47(3):982--1038, 2018.

\bibitem[BB20]{DBLP:conf/focs/Ben-DavidB20}
Shalev Ben{-}David and Eric Blais.
\newblock A tight composition theorem for the randomized query complexity of
  partial functions: Extended abstract.
\newblock In {\em {FOCS}}, pages 240--246. {IEEE}, 2020.

\bibitem[Bon70]{Bonami70}
Aline Bonami.
\newblock Étude des coefficients de fourier des fonctions de $l^p(g)$.
\newblock {\em Annales de l'institut Fourier}, 20(2):335--402, 1970.
\newblock URL: \url{http://eudml.org/doc/74019}.

\bibitem[BP18]{briggs2018extremal}
Joseph Briggs and Wesley Pegden.
\newblock Extremal collections of $ k $-uniform vectors.
\newblock {\em arXiv preprint arXiv:1801.09609}, 2018.

\bibitem[BS20]{DBLP:journals/eccc/BansalS20}
Nikhil Bansal and Makrand Sinha.
\newblock {\textdollar}k{\textdollar}-forrelation optimally separates quantum
  and classical query complexity.
\newblock {\em Electron. Colloquium Comput. Complex.}, 27:127, 2020.

\bibitem[BTW15]{DBLP:journals/corr/BlaisTW15}
Eric Blais, Li{-}Yang Tan, and Andrew Wan.
\newblock An inequality for the fourier spectrum of parity decision trees.
\newblock {\em CoRR}, abs/1506.01055, 2015.

\bibitem[CFK{\etalchar{+}}19]{DBLP:conf/icalp/ChattopadhyayFK19}
Arkadev Chattopadhyay, Yuval Filmus, Sajin Koroth, Or~Meir, and Toniann
  Pitassi.
\newblock Query-to-communication lifting for {BPP} using inner product.
\newblock In Christel Baier, Ioannis Chatzigiannakis, Paola Flocchini, and
  Stefano Leonardi, editors, {\em 46th International Colloquium on Automata,
  Languages, and Programming, {ICALP} 2019, July 9-12, 2019, Patras, Greece},
  volume 132 of {\em LIPIcs}, pages 35:1--35:15. Schloss Dagstuhl -
  Leibniz-Zentrum f{\"{u}}r Informatik, 2019.
\newblock \href {https://doi.org/10.4230/LIPIcs.ICALP.2019.35}
  {\path{doi:10.4230/LIPIcs.ICALP.2019.35}}.

\bibitem[CGL{\etalchar{+}}20]{DBLP:journals/corr/abs-2008-01316}
Eshan Chattopadhyay, Jason Gaitonde, Chin~Ho Lee, Shachar Lovett, and Abhishek
  Shetty.
\newblock Fractional pseudorandom generators from any fourier level.
\newblock {\em CoRR}, abs/2008.01316, 2020.

\bibitem[CHHL19]{DBLP:journals/toc/ChattopadhyayHH19}
Eshan Chattopadhyay, Pooya Hatami, Kaave Hosseini, and Shachar Lovett.
\newblock Pseudorandom generators from polarizing random walks.
\newblock {\em Theory Comput.}, 15:1--26, 2019.

\bibitem[CHLT19]{DBLP:conf/innovations/ChattopadhyayHL19}
Eshan Chattopadhyay, Pooya Hatami, Shachar Lovett, and Avishay Tal.
\newblock Pseudorandom generators from the second fourier level and
  applications to {AC0} with parity gates.
\newblock In {\em {ITCS}}, volume 124 of {\em LIPIcs}, pages 22:1--22:15.
  Schloss Dagstuhl - Leibniz-Zentrum f{\"{u}}r Informatik, 2019.

\bibitem[CHRT18]{DBLP:conf/stoc/ChattopadhyayHR18}
Eshan Chattopadhyay, Pooya Hatami, Omer Reingold, and Avishay Tal.
\newblock Improved pseudorandomness for unordered branching programs through
  local monotonicity.
\newblock In {\em {STOC}}, pages 363--375. {ACM}, 2018.

\bibitem[CPT20]{DBLP:journals/eccc/CohenPT20}
Gil Cohen, Noam Peri, and Amnon Ta{-}Shma.
\newblock Expander random walks: {A} fourier-analytic approach.
\newblock {\em Electron. Colloquium Comput. Complex.}, 27:163, 2020.

\bibitem[CS16]{CS16}
Gil Cohen and Igor Shinkar.
\newblock The complexity of {DNF} of parities.
\newblock In {\em {ITCS}}, pages 47--58. {ACM}, 2016.

\bibitem[Gav16]{gavinsky}
Dmitry Gavinsky.
\newblock Entangled simultaneity versus classical interactivity in
  communication complexity.
\newblock In {\em Proceedings of the Forty-Eighth Annual ACM Symposium on
  Theory of Computing}, STOC '16, page 877–884, New York, NY, USA, 2016.
  Association for Computing Machinery.
\newblock \href {https://doi.org/10.1145/2897518.2897545}
  {\path{doi:10.1145/2897518.2897545}}.

\bibitem[GRT21]{DBLP:conf/innovations/GirishRT21}
Uma Girish, Ran Raz, and Avishay Tal.
\newblock Quantum versus randomized communication complexity, with efficient
  players.
\newblock In {\em {ITCS}}, volume 185 of {\em LIPIcs}, pages 54:1--54:20.
  Schloss Dagstuhl - Leibniz-Zentrum f{\"{u}}r Informatik, 2021.

\bibitem[GRZ20]{DBLP:journals/eccc/GirishRZ20a}
Uma Girish, Ran Raz, and Wei Zhan.
\newblock Lower bounds for {XOR} of forrelations.
\newblock {\em Electron. Colloquium Comput. Complex.}, 27:101, 2020.

\bibitem[GSTW16]{DBLP:journals/eccc/GopalanSTW16}
Parikshit Gopalan, Rocco~A. Servedio, Avishay Tal, and Avi Wigderson.
\newblock Degree and sensitivity: tails of two distributions.
\newblock {\em Electron. Colloquium Comput. Complex.}, 23:69, 2016.

\bibitem[HHL18]{DBLP:journals/siamcomp/HatamiHL18}
Hamed Hatami, Kaave Hosseini, and Shachar Lovett.
\newblock Structure of protocols for {XOR} functions.
\newblock {\em {SIAM} J. Comput.}, 47(1):208--217, 2018.

\bibitem[KM93]{DBLP:journals/siamcomp/KushilevitzM93}
Eyal Kushilevitz and Yishay Mansour.
\newblock Learning decision trees using the fourier spectrum.
\newblock {\em {SIAM} J. Comput.}, 22(6):1331--1348, 1993.

\bibitem[KQS15]{KQS15}
Raghav Kulkarni, Youming Qiao, and Xiaoming Sun.
\newblock On the power of parity queries in boolean decision trees.
\newblock In {\em {TAMC}}, volume 9076 of {\em Lecture Notes in Computer
  Science}, pages 99--109. Springer, 2015.

\bibitem[Kra10]{DBLP:journals/combinatorics/Kramer10}
Joshua~Brown Kramer.
\newblock On the most weight \emph{w} vectors in a dimension \emph{k} binary
  code.
\newblock {\em Electron. J. Comb.}, 17(1), 2010.

\bibitem[Lee19]{DBLP:conf/coco/Lee19}
Chin~Ho Lee.
\newblock Fourier bounds and pseudorandom generators for product tests.
\newblock In {\em Computational Complexity Conference}, volume 137 of {\em
  LIPIcs}, pages 7:1--7:25. Schloss Dagstuhl - Leibniz-Zentrum f{\"{u}}r
  Informatik, 2019.

\bibitem[Man95]{DBLP:journals/jcss/Mansour95}
Yishay Mansour.
\newblock An o(n{\^{}}(log log n)) learning algorithm for {DNF} under the
  uniform distribution.
\newblock {\em J. Comput. Syst. Sci.}, 50(3):543--550, 1995.

\bibitem[MNR11]{DBLP:journals/tcs/MontanaroNR11}
Ashley Montanaro, Harumichi Nishimura, and Rudy Raymond.
\newblock Unbounded-error quantum query complexity.
\newblock {\em Theor. Comput. Sci.}, 412(35):4619--4628, 2011.
\newblock \href {https://doi.org/10.1016/j.tcs.2011.04.043}
  {\path{doi:10.1016/j.tcs.2011.04.043}}.

\bibitem[MO09]{MO09}
Ashley Montanaro and Tobias Osborne.
\newblock On the communication complexity of {XOR} functions.
\newblock {\em CoRR}, abs/0909.3392, 2009.

\bibitem[MS20]{MS20}
Nikhil~S. Mande and Swagato Sanyal.
\newblock On parity decision trees for fourier-sparse boolean functions.
\newblock In {\em {FSTTCS}}, volume 182 of {\em LIPIcs}, pages 29:1--29:16.
  Schloss Dagstuhl - Leibniz-Zentrum f{\"{u}}r Informatik, 2020.

\bibitem[NN93]{DBLP:journals/siamcomp/NaorN93}
Joseph Naor and Moni Naor.
\newblock Small-bias probability spaces: Efficient constructions and
  applications.
\newblock {\em {SIAM} J. Comput.}, 22(4):838--856, 1993.

\bibitem[O'D12]{DBLP:journals/corr/abs-1204-6447}
Ryan O'Donnell.
\newblock Open problems in analysis of boolean functions.
\newblock {\em CoRR}, abs/1204.6447, 2012.

\bibitem[O'D14]{DBLP:books/daglib/0033652}
Ryan O'Donnell.
\newblock {\em Analysis of Boolean Functions}.
\newblock Cambridge University Press, 2014.

\bibitem[OS07]{DBLP:journals/siamcomp/ODonnellS07}
Ryan O'Donnell and Rocco~A. Servedio.
\newblock Learning monotone decision trees in polynomial time.
\newblock {\em {SIAM} J. Comput.}, 37(3):827--844, 2007.

\bibitem[OWZ{\etalchar{+}}14]{OWZST14}
Ryan O'Donnell, John Wright, Yu~Zhao, Xiaorui Sun, and Li{-}Yang Tan.
\newblock A composition theorem for parity kill number.
\newblock In {\em Computational Complexity Conference}, pages 144--154. {IEEE}
  Computer Society, 2014.

\bibitem[RSV13]{DBLP:conf/approx/ReingoldSV13}
Omer Reingold, Thomas Steinke, and Salil~P. Vadhan.
\newblock Pseudorandomness for regular branching programs via fourier analysis.
\newblock In {\em {APPROX-RANDOM}}, volume 8096 of {\em Lecture Notes in
  Computer Science}, pages 655--670. Springer, 2013.

\bibitem[San19]{Sanyal19}
Swagato Sanyal.
\newblock Fourier sparsity and dimension.
\newblock {\em Theory Comput.}, 15:1--13, 2019.

\bibitem[SSW20]{DBLP:journals/eccc/SherstovSW20}
Alexander~A. Sherstov, Andrey~A. Storozhenko, and Pei Wu.
\newblock An optimal separation of randomized and quantum query complexity.
\newblock {\em Electron. Colloquium Comput. Complex.}, 27:128, 2020.

\bibitem[STlV17]{STV17}
Amir Shpilka, Avishay Tal, and Ben lee Volk.
\newblock On the structure of boolean functions with small spectral norm.
\newblock {\em Comput. Complex.}, 26(1):229--273, 2017.

\bibitem[Tal17]{DBLP:conf/coco/Tal17}
Avishay Tal.
\newblock Tight bounds on the fourier spectrum of {AC0}.
\newblock In {\em Computational Complexity Conference}, volume~79 of {\em
  LIPIcs}, pages 15:1--15:31. Schloss Dagstuhl - Leibniz-Zentrum f{\"{u}}r
  Informatik, 2017.

\bibitem[Tal20]{DBLP:conf/focs/Tal20}
Avishay Tal.
\newblock Towards optimal separations between quantum and randomized query
  complexities.
\newblock In {\em {FOCS}}, pages 228--239. {IEEE}, 2020.

\bibitem[TWXZ13]{TWXZ13}
Hing~Yin Tsang, Chung~Hoi Wong, Ning Xie, and Shengyu Zhang.
\newblock Fourier sparsity, spectral norm, and the log-rank conjecture.
\newblock In {\em {FOCS}}, pages 658--667. {IEEE} Computer Society, 2013.

\bibitem[ZS09]{ZS09}
Zhiqiang Zhang and Yaoyun Shi.
\newblock Communication complexities of symmetric {XOR} functions.
\newblock {\em Quantum Inf. Comput.}, 9(3{\&}4):255--263, 2009.

\bibitem[ZS10]{ZS10}
Zhiqiang Zhang and Yaoyun Shi.
\newblock On the parity complexity measures of boolean functions.
\newblock {\em Theor. Comput. Sci.}, 411(26-28):2612--2618, 2010.

\end{thebibliography}

\appendix

\section[Proof of Corollary 1.8]{Proof of \Cref{cor:small-size}}\label{sec:small-size decision trees}
\begin{corollary*}[\Cref{cor:small-size} restated]
Let $\Tcal$ be a parity decision tree of size at most $s>1$ on $n$ variables.
Then, 
$$\forall{\ell \in [n]}: L_{1,\ell}(f) \le (\log(s))^{\ell/2} \cdot O(\ell \cdot \log(n))^{1.5\ell}.$$
\end{corollary*}
\begin{proof}
	We approximate $\Tcal$ with error $\eps = 1/n^{\ell}$ by another parity decision tree $\Tcal'$ of depth $d =\lceil\log\pbra{s \cdot n^{\ell}}\rceil$, where we simply replace all nodes of depth $d$ in $\Tcal$ with leaves that return $0$.
	Since there are at most $s$ nodes in $\Tcal$, the probability that a random input would reach one of the nodes of depth $d$ is at most $2^{-d} \cdot s \le 1/n^{\ell}$. 
	Hence $\Pr_{x}\sbra{\Tcal(x) \neq \Tcal'(x)} \le \eps$.
	This implies that $\abs{\hat{\Tcal}(S) - \hat{\Tcal'}(S)} \le \eps$ for any subset $S \subseteq [n]$.
	Thus, 
	$$
	L_{1,\ell}(\Tcal) = \sum_{S:|S|=\ell} \abs{\hat{\Tcal}(S)} \le \sum_{S:|S|=\ell} \pbra{\abs{\hat{\Tcal'}(S)} +\eps} \le L_{1,\ell}(\Tcal') + 1.
	$$
	Since $\Tcal'$ is of depth at most $d =\lceil\log(s) + \ell \cdot \log (n)\rceil=O\pbra{\log(s) \cdot\ell \cdot \log(n)}$, we obtain our bound.
\end{proof}

\section[Proof of Lemma 3.3]{Proof of \Cref{lem:adaptive_azuma_new}}\label{app:lem:adaptive_azuma_new}

We will use the definition of sub-Gaussian random variables.
\begin{definition}[Sub-Gaussian random variables]
We say a random variable $x$ is $\Delta$-sub-Gaussian if $\E\sbra{e^{t\cdot x}}\le e^{t^2\Delta^2}$ holds for all $t\in\Rbb$.
\end{definition}

Now we prove the following sub-Gaussian adaptive Azuma's inequality.
\begin{lemma}[Sub-Gaussian adaptive Azuma's inequality]\label{lem:sub_gaussian_azuma}
Let $X^{(0)},\ldots,X^{(D)}$ be a martingale with respect to a filtration $\pbra{\Fcal^{(i)}}_{i=0}^D$\footnote{$\Fcal^{(0)}\subseteq\Fcal^{(1)}\subseteq\cdots\subseteq\Fcal^{(D)}$ is an increasing sequence of $\sigma$-algebra where each $\Fcal^{(i)}$ makes $X^{(0)},\ldots,X^{(i+1)}$ measurable and $\E\sbra{X^{(i)}\mid\Fcal^{(i-1)}}=X^{(i-1)}$. Intuitively, the filtration is the history of the martingale.} and $\Delta^{(1)},\ldots,\Delta^{(D)}$ be a sequence of magnitudes such that $X^{(0)}=0$ and $X^{(i)}=X^{(i-1)}+\delta^{(i)}$ for $i\in[D]$, where if conditioning on $\Fcal^{(i-1)}$, $\delta^{(i)}$ is a $\Delta^{(i)}$-sub-Gaussian random variable and $\Delta^{(i)}$ is a fixed value.

If there exists some constant $U\ge0$ such that $\sum_{i=1}^D\abs{\Delta^{(i)}}^2\le U$ always holds, then for any $\beta\ge0$ we have
$$
\Pr\sbra{\max_{i=0,1,\ldots,D}\abs{X^{(i)}}\ge\beta\cdot\sqrt{2U}}\le 2\cdot e^{-\beta^2/2}.
$$
\end{lemma}
\begin{proof}
The bound holds trivially when $\beta=0$, hence we assume $\beta>0$ from now on.
We construct another martingale ${\hat X}^{(0)},\ldots,{\hat X}^{(D)}$ as follows:
$$
{\hat X}^{(i)}=\begin{cases}
X^{(i)} & 0\le i\le d,\\
X^{(d)} & i>d,
\end{cases}
\quad\text{where}\quad
d=\min\cbra{D}\cup\cbra{i\in\cbra{0,1\ldots,D}\mid \abs{X^{(i)}}\ge\beta\cdot\sqrt{2U}}.
$$
We write $\hat{\delta}^{(i)}=\hat{X}^{(i)}-{\hat X}^{(i-1)}$, then $\hat{\delta}^{(i)}=\delta^{(i)}$ for all $i\le d$; and $\hat{\delta}^{(i)}\equiv0$ for all $i>d$.
Let $\hat{\Delta}^{(i)}=\Delta^{(i)}$ for all $i\le d$; and $\hat{\Delta}^{(i)}\equiv0$ for all $i>d$.
Thus $\hat{\delta}^{(i)}$ is $\hat{\Delta}^{(i)}$-sub-Gaussian given $\Fcal^{(i-1)}$; and 
$$
\sum_{i=1}^D\abs{\hat{\Delta}^{(i)}}^2=\sum_{i=1}^d\abs{\Delta^{(i)}}^2\le U.
$$
Moreover, we have
$$
\Pr\sbra{\max_{i=0,1,\ldots,D}\abs{X^{(i)}}\ge\beta\cdot\sqrt{2U}}=\Pr\sbra{\abs{{\hat X}^{(D)}}\ge\beta\cdot\sqrt{2U}}.
$$ 
Let $t>0$ be a parameter and we bound $\E\sbra{e^{t\cdot\hat{X}^{(D)}}}$ as follows
\begingroup
\allowdisplaybreaks
\begin{align}
\E\sbra{e^{t\cdot\hat{X}^{(D)}}}
&=\E_{\Fcal^{(D-1)}}\sbra{e^{t\cdot\hat{X}^{(D-1)}}\cdot\E_{\Fcal^{(D)}}\sbra{e^{t\cdot\pbra{\hat{X}^{(D)}-\hat{X}^{(D-1)}}}\mid\Fcal^{(D-1)}}}
\label{eq:clm:adaptive_azuma_1}\\
&=\E_{\Fcal^{(D-1)}}\sbra{e^{t\cdot\hat{X}^{(D-1)}}\cdot\E_{\Fcal^{(D)}}\sbra{e^{t\cdot\hat{\delta}^{(D)}}\mid\Fcal^{(D-1)}}}
\label{eq:clm:adaptive_azuma_2}\\
&\le\E_{\Fcal^{(D-1)}}\sbra{e^{t\cdot\hat{X}^{(D-1)}}\cdot e^{t^2\pbra{\hat{\Delta}^{(D)}}^2}}
\tag{since $\hat{\delta}^{(D)}$ is $\hat{\Delta}^{(D)}$-sub-Gaussian}\\
&\le\E_{\Fcal^{(D-1)}}\sbra{e^{t\cdot\hat{X}^{(D-1)}}\cdot e^{t^2\pbra{U-\pbra{\hat{\Delta}^{(1)}}^2-\cdots-\pbra{\hat{\Delta}^{(D-1)}}^2}}}
\notag\\
&\le\E_{\Fcal^{(D-2)}}\sbra{e^{t\cdot\hat{X}^{(D-2)}}\cdot e^{t^2\pbra{U-\pbra{\hat{\Delta}^{(1)}}^2-\cdots-\pbra{\hat{\Delta}^{(D-1)}}^2}}e^{t^2\pbra{\hat{\Delta}^{(D-1)}}^2}}
\tag{similar to \Cref{eq:clm:adaptive_azuma_1} and \Cref{eq:clm:adaptive_azuma_2}}\\
&=\E_{\Fcal^{(D-2)}}\sbra{e^{t\cdot\hat{X}^{(D-2)}}\cdot e^{t^2\pbra{U-\pbra{\hat{\Delta}^{(1)}}^2-\cdots-\pbra{\hat{\Delta}^{(D-2)}}^2}}}
\notag\\
&\le\cdots
\le\E_{\Fcal^{(D-k)}}\sbra{e^{t\cdot\hat{X}^{(D-k)}}\cdot e^{t^2\pbra{U-\pbra{\hat{\Delta}^{(1)}}^2-\cdots-\pbra{\hat{\Delta}^{(D-k)}}^2}}}
\le\cdots\notag\\
&\le e^{t^2U}.
\end{align}
\endgroup
Setting $t=\beta/\sqrt{2U}$ implies that
$$
\Pr\sbra{\hat{X}^{(D)}\ge \beta \cdot \sqrt{2U}} \le \frac{\E\sbra{e^{t\cdot\hat{X}^{(D)}}}}{e^{t\cdot\beta\cdot\sqrt{2U} }} \le\frac{e^{t^2U}}{e^{\beta^2}}= e^{-\beta^2/2}. 
$$
Similarly we can show $\Pr\sbra{\hat{X}^{(D)}\le -\beta \cdot \sqrt{2U}} \le e^{-\beta^2/2}$, which completes the proof by a union bound.
\end{proof}

For our applications, we need the following fact.
\begin{fact}\label{fct:bounded_implies_sub_gaussian}
Let $x$ be a mean-zero random variable and assume $|x|\le\Delta$ always holds. Then $x$ is $\Delta$-sub-Gaussian.
\end{fact}
\begin{proof}
Note that $e^{t\cdot x}$ is convex for all $t\in\Rbb$. By Jensen's inequality, we have
\begin{equation*}
\E\sbra{e^{t\cdot x}}\le\frac12\pbra{e^{-t\Delta}+e^{t\Delta}}=\sum_{i=0}^{+\infty}\frac{\pbra{t\Delta}^{2i}}{(2i)!}\le\sum_{i=0}^{+\infty}\frac{\pbra{t\Delta}^{2i}}{i!}=e^{t^2\Delta^2}.
\tag*{\qedhere}
\end{equation*}
\end{proof}

As a corollary of \Cref{lem:sub_gaussian_azuma} and \Cref{fct:bounded_implies_sub_gaussian}, we obtain \Cref{lem:adaptive_azuma_new}.
\begin{corollary*}[\Cref{lem:adaptive_azuma_new} restated]
Let $X^{(0)},\ldots,X^{(D)}$ be a martingale and $\Delta^{(1)},\ldots,\Delta^{(D)}$ be a sequence of magnitudes such that $X^{(0)}=0$ and $X^{(i)}=X^{(i-1)}+\Delta^{(i)}\cdot z^{(i)}$ for $i\in[D]$, where if conditioning on $z^{(1)},\ldots,z^{(i-1)}$,
\begin{itemize}
\item[(1)] $z^{(i)}$ is a mean-zero random variable and $\abs{z^{(i)}}\le1$ always holds;
\item[(2)] $\Delta^{(i)}$ is a fixed value.
\end{itemize}
If there exists some constant $U\ge0$ such that $\sum_{i=1}^D\abs{\Delta^{(i)}}^2\le U$ always holds, then for any $\beta\ge0$ we have
$$
\Pr\sbra{\max_{i=0,1,\ldots,D}\abs{X^{(i)}}\ge\beta\cdot\sqrt{2U}}\le 2\cdot e^{-\beta^2/2}.
$$
\end{corollary*}

\section[Proof of Claim 5.8]{Proof of \Cref{clm:level_l_Gamma}}\label{app:clm:level_l_Gamma}

\begin{claim*}[\Cref{clm:level_l_Gamma} restated]
$\Pr\sbra{\Ecal_2}\le\eps/3$, where $\Ecal_2$ is the following event:
$$
\text{`` }\exists T\in\binom{[n]}{\ell-t},i,r,r',\abs{\Gamma_T^{(i)}(r,r')}\ge\pbra{100\min\cbra{k,\log\pbra{\tfrac{n^\ell}\eps}}\cdot\pbra{\tfrac{n^\ell}\eps}^{\frac6k}}^{\frac{r+r'}2}\cdot\sigma_T(r,r',C(v_{i-1}),i)\text{''}.
$$
\end{claim*}
\begin{proof}
Let $k'=\min\cbra{k,\left\lceil6\log\pbra{n^\ell/\eps}\right\rceil}\le12\min\cbra{k,\log\pbra{n^\ell/\eps}}$. Then $\Tcal$ is also a depth-$D$ $2k'$-clean parity decision tree.
Observe that
\begin{align*}
&\phantom{\le}\Pr\sbra{\abs{\Gamma_T^{(i)}(r,r')}\ge
\pbra{\frac{4k'}{\eta^{2/k'}}}^{(r+r')/2}
\cdot\sigma_T(r,r',C(v_{i-1}),i)}\\
&\le\max_{C(v_{i-1})}
\Pr\sbra{\abs{\Gamma_T^{(i)}(r,r')}\ge
\pbra{\frac{4k'}{\eta^{2/k'}}}^{(r+r')/2}
\cdot\sigma_T(r,r',C(v_{i-1}),i)\mid C(v_{i-1})}\\
&\le\underbrace{\frac{(4\cdot k')^{r+r'}}{(2\cdot(r+r'))^{k'}}}_{\le1}\cdot\underbrace{\eta^{2-\frac{2(r+r')}{k'}}}_{\le\eta}
\tag{due to the second bound in \Cref{lem:low_degree_polys} and $k\ge4\cdot\ell\ge4\cdot(r+r')$}\\
&\le\eta.
\end{align*}
Thus by union bound over all $T\in\binom{[n]}{\ell-t},i\in[D'],r\in[t],0\le r'\le t-r$, we have
$$
\Pr\sbra{\exists T,i,r,r',~\abs{\Gamma_T^{(i)}(r,r')}\ge
\pbra{\tfrac{4k}{\eta^{2/k}}}^{(r+r')/2}
\cdot\sigma_T(r,r',C(v_{i-1}),i)}
\le Dt^2n^{\ell-t}\cdot\eta
\le \tfrac{n^{\ell+2}\cdot\eta}3
\le \tfrac{n^{3\cdot\ell}\cdot\eta}3,
$$
where we use the fact $n\ge\max\cbra{D,3\cdot t}$ and $t\ge1$.
By setting $\eta=\eps/n^{3\cdot\ell}$, we have
\begin{align*}
\frac{4k'}{\eta^{2/k'}}=4k'\pbra{\frac{n^{3\cdot\ell}}\eps}^{\frac2{k'}}\le4k'\pbra{\frac{n^\ell}\eps}^{\frac6{k'}}\le4\cdot12\min\cbra{k,\log\pbra{\frac{n^\ell}\eps}}\cdot2\pbra{\frac{n^\ell}\eps}^{\frac6k},
\end{align*}
as desired.
\end{proof}

\section[Proof of Claim 5.10]{Proof of \Cref{clm:level_l_final}}\label{app:clm:level_l_final}

We first need the following simple bound on $M$.
\begin{lemma}\label{lem:M_sums}
For any integer $s\ge1$, we have
$$
\sum_{r=s}^tM(D,d,k,\ell,t-r,\eps)\le\frac{2\cdot M(D,d,k,\ell,t,\eps)}{\pbra{\tau D\cdot\log\pbra{n^\ell/\eps}}^{s/2}}.
$$
\end{lemma}
\begin{proof}
We simply expand the formula of $M$ as follows:
\begin{align*}
\frac{\sum_{r=s}^tM(D,d,k,\ell,t-r,\eps)}{M(D,d,k,\ell,t,\eps)}
&=\sum_{r=s}^t\pbra{\tau \cdot (D+dk)\cdot\pbra{\tfrac{n^\ell}\eps}^{6/k}\log\pbra{\tfrac{n^\ell}\eps}}^{-r/2}\\
&\le\sum_{r=s}^{+\infty}\pbra{\tau \cdot (D+dk)\cdot\pbra{\tfrac{n^\ell}\eps}^{6/k}\log\pbra{\tfrac{n^\ell}\eps}}^{-r/2}\\
&\le2\cdot\pbra{\tau \cdot (D+dk)\cdot\pbra{\tfrac{n^\ell}\eps}^{6/k}\log\pbra{\tfrac{n^\ell}\eps}}^{-s/2}
\tag{due to $\tau\ge 4$ and $s\ge1$}\\
&\le2\cdot\pbra{\tau D\cdot\log\pbra{n^\ell/\eps}}^{-s/2}.
\tag*{\qedhere}
\end{align*}
\end{proof}
Now we prove \Cref{clm:level_l_final}.
\begin{claim*}[\Cref{clm:level_l_final} restated]
When $\Ecal_1\lor\Ecal_2$ does not happen, $\sum_{i=1}^D\abs{\mu_T^{(i)}}\le R$ and $\sum_{i=1}^D\abs{\delta_T^{(i)}}^2\le R^2$ hold for all $T\in\binom{[n]}{\ell-t}$, where 
$$
R=\frac{M(D,d,k,\ell,t,\eps)}{5\cdot\sqrt{\log\pbra{n^\ell/\eps}}}.
$$
\end{claim*}
\begin{proof}
We verify for each $T\in\binom{[n]}{\ell-t}$ as follows:
\begin{align*}
\sum_{i=1}^D\abs{\mu_T^{(i)}}
&=\sum_{i=1}^{D'}\abs{\mu_T^{(i)}}\le\sum_{i=1}^{D'}\sum_{\substack{r=2,\\\text{even}}}^t\abs{A(T,r,i)}
\tag{due to \Cref{eq:def_mu_delta}}\\
&\le\sum_{i=1}^{D'}\sum_{\substack{r=2,\\\text{even}}}^t
\pbra{M(D,d,k,\ell,t-r,\eps)+\tfrac{M(D,d,k,\ell,t,\eps)}{\sqrt{\log\pbra{n^\ell/\eps}}}\cdot\sqrt{\pbra{\tfrac{800}\tau}^r\pbra{\tfrac{\abs{J(v_{i-1})}}{2d}}^r} \cdot \indicator_{|J(v_{i-1})|>1}}
\tag{due to \Cref{eq:final_bound_on_A}}\\
&\le\sum_{i=1}^{D'}\sum_{\substack{r=2,\\\text{even}}}^t
\pbra{M(D,d,k,\ell,t-r,\eps)+\tfrac{M(D,d,k,\ell,t,\eps)}{\sqrt{\log\pbra{n^\ell/\eps}}}\cdot\pbra{\tfrac{\abs{J(v_{i-1})}}{2d}}\pbra{\tfrac{800}\tau}^{r/2}\cdot \indicator_{|J(v_{i-1})|>1}}
\tag{Since $|J(v_{i-1})| \le 2d$ from \Cref{cor:depth_to_singleton}}\\
&\le\tfrac{2\cdot M(D,d,k,\ell,t,\eps)}{\tau\cdot\log\pbra{n^\ell/\eps}}+\tfrac{1.1\cdot 800 \cdot M(D,d,k,\ell,t,\eps)}{\tau \cdot\sqrt{\log\pbra{n^\ell/\eps}}}
\tag{due to \Cref{lem:M_sums} and \Cref{cor:depth_to_singleton} and $\tau = 10^4$}\\
&\le\tfrac{M(D,d,k,\ell,t,\eps)}{5\cdot\sqrt{\log\pbra{n^\ell/\eps}}}=R
\end{align*}
and with similar calculation, we have
\begin{align*}
\sum_{i=1}^D\abs{\delta_T^{(i)}}^2
&\le\sum_{i=1}^{D'}\pbra{\sum_{\substack{r=1,\\\text{odd}}}^t\pbra{
M(D,d,k,\ell,t-r,\eps)+\tfrac{M(D,d,k,\ell,t,\eps)}{\sqrt{\log\pbra{n^\ell/\eps}}}\cdot\sqrt{\tfrac{\abs{J(v_{i-1})}}{2d}}\pbra{\tfrac{800}\tau}^{r/2} \cdot \indicator_{|J(v_{i-1})|>1}} }^2\\
&\le\sum_{i=1}^{D'}\pbra{\tfrac{2\cdot M(D,d,k,\ell,t,\eps)}{\sqrt{\tau D\cdot\log\pbra{n^\ell/\eps}}}
+\tfrac{1.1\cdot\sqrt{800}\cdot M(D,d,k,\ell,t,\eps)}{\sqrt{\tau} \sqrt{\log\pbra{n^\ell/\eps}}}\cdot\sqrt{\tfrac{\abs{J(v_{i-1})}}{2d}} \cdot \indicator_{|J(v_{i-1})|>1}}^2\tag{due to $\tau = 10^4$}\\
&\le\pbra{\tfrac{M(D,d,k,\ell,t,\eps)}{\sqrt{\log\pbra{n^\ell/\eps}}}}^2\sum_{i=1}^{D'}2\cdot\pbra{\frac4{\tau D}+\tfrac{968}{\tau}\cdot\tfrac{\abs{J(v_{i-1})}}{2d}  \cdot \indicator_{|J(v_{i-1})|>1}}
\tag{due to $(a+b)^2\le2(a^2+b^2)$}\\
&\le\pbra{ \tfrac{2000 \cdot M(D,d,k,\ell,t,\eps)}{\tau \cdot\sqrt{\log\pbra{n^\ell/\eps}}}}^2=R^2.
\tag*{\qedhere}
\end{align*}
\end{proof}

\end{document}